\documentclass[pdflatex,sn-mathphys-num]{sn-jnl}

\usepackage{graphicx}%
\usepackage{multirow}%
\usepackage{amsmath,amssymb,amsfonts}%
\usepackage{amsthm}%
\usepackage{mathrsfs}%
\usepackage[title]{appendix}%
\usepackage{xcolor}%
\usepackage{textcomp}%
\usepackage{manyfoot}%
\usepackage{booktabs}%
\usepackage{algorithm}%
\usepackage{algorithmicx}%
\usepackage{algpseudocode}%
\usepackage{listings}%

\usepackage{amssymb}
\usepackage{enumerate}
\usepackage{graphicx}
\usepackage{amsmath,mathrsfs}
\usepackage{multicol}
\usepackage{tikz}
\usetikzlibrary{automata,positioning}
\usepackage{wrapfig}
\usepackage{bussproofs}
\usepackage{times}
\usepackage{color}
\usepackage[all]{xy}
\usepackage{imakeidx}

\newtheorem{thm}{Theorem}[section]
 \newtheorem{cor}[thm]{Corollary}
 \newtheorem{lem}[thm]{Lemma}
 \newtheorem{prop}[thm]{Proposition}

 \theoremstyle{definition}
 \newtheorem{defn}[thm]{Definition}

 \theoremstyle{remark}
 
 \newtheorem*{ex}{Example}
 \numberwithin{equation}{section}

\newcommand{\ra}{\Rightarrow}

\newcommand{\naraba}{\rightarrow}

\newcommand{\inset}[2]{\left\{\, {#1} \,|\, {#2} \,\right\}}
\newcommand{\setof}[1]{\{{#1}\}}

\newcommand{\md}{\mathbb{M}}
\newcommand{\mdd}{\mathbb{M}_{\mathbf{ser }}}
\newcommand{\mdt}{\mathbb{M}_{\mathbf{ref}}}

\newcommand{\hkd}{\mathsf{H}(\mathbf{K}_D) }
\newcommand{\gkd}{\mathsf{G}\mathbf{(K}_D) }

\newcommand{\hkdd}{\mathsf{H}\mathbf{(KD}_D) }
\newcommand{\gkdd}{\mathsf{G}\mathbf{(KD}_D) }

\newcommand{\hktd}{\mathsf{H} \mathbf{(KT}_D) }
\newcommand{\gktd}{\mathsf{G}\mathbf{(KT}_D) }

\newcommand{\gktdplus}{\mathsf{G}\mathbf{(KT}^+_D) }

\begin{document}


\title{Uniform
Agent-interpolation  
of Distributed Knowledge
}


\author*[1]{\fnm{Youan} \sur{Su}}\email{su.youan@lnu.edu.cn}

\affil*[1]{\orgdiv{School of Philosophy}, \orgname{Liaoning University}, \orgaddress{\street{No.58, Daoyinan Street}, \city{Shenyang}, \country{China}}}


\abstract{ 



Uniform interpolation property (UIP) is a strengthening of Craig interpolation property. 
It can be understood as the definability of propositional quantifiers.
This paper develops the sequent calculi provided in Murai and Sano (2020), combining  with the methods studied by B{\'\i}lkov{\'a} (2007) to show the uniform interpolation for
epistemic logic   $\mathbf{K}$, $\mathbf{KD}$ and  $\mathbf{KT}$ with distributed knowledge. 
 A purely syntactic algorithm is presented to determine a uniform interpolant formula. In the definition of an interpolant formula, not only propositional variables but also agent symbols are  taken into consideration.
}



\maketitle

\section{introduction}

Craig interpolation theorem  states that: 
 whenever $\vdash A \rightarrow B$, there exists an interpolant formula $C$ such that  $\vdash A \rightarrow C$ and $\vdash C \rightarrow B$, and all propositional variables appearing in $C$ are shared by $A$ and $B$.
 The  Uniform Interpolation Property (UIP)  is a strengthening of Craig interpolation, first established  by Pitts   \cite{pitts1992interpretation}. 
The existence of uniform interpolants can also be viewed as the possibility of a simulation of certain propositional quantifiers.
For an arbitrary formula $A$ and any propositional variable $p_1,\ldots,p_m$, there is a post-interpolant $\exists p_1,\ldots,p_m(A)$ not containing $p_1,\ldots,p_m$, such that
$\vdash A \rightarrow \exists p_1,\ldots,p_m(A)$, and $\vdash A \rightarrow B$ implies $\vdash \exists p_1,\ldots,p_m(A) \rightarrow B$ for each $B$ not containing $p_1,\ldots,p_m$.
Similarly, for any $B$ and $p_1,\ldots,p_n$ there is a pre-interpolant $\forall p_1,\ldots,p_n(B)$ not containing $p_1,\ldots,p_n$, such
that $\vdash \forall p_1,\ldots,p_n(B) \rightarrow B$ and $\vdash A \rightarrow B$ implies $\vdash A \rightarrow \forall p_1,\ldots,p_n(B)$ for each $A$ not containing
$p_1,\ldots,p_n$.

Pitts \cite{pitts1992interpretation} first  established the UIP as a strengthening of Craig interpolation for intuitionistic propositional logic, based on a sequent calculus that absorbs all structural rules.
In recent years, UIP has gained much attention in the study of modal logic. 
UIP in modal logic $\mathbf{K}$ was shown by Visser \cite{Visser1996} using bounded bisimulations and by Ghilardi \cite{Ghilardi1995} using an algebraic approach.
Wolter \cite{wolter1997} proved that modal logic $\mathbf{S5}$ has the UIP.
It is also known that $\mathbf{K4}$ and $\mathbf{S4}$ do not satisfy UIP \cite{Ghilardi1995,bilkova2007uniform}.
Regarding multi-agent modal logic, Wolter \cite{wolter1997} showed that UIP for any normal mono-modal logic can be generalized to its multi-agent case.
Fang et al. \cite{FANG201951} proved that $\mathbf{K_n}$
$\mathbf{D_n}$,
$\mathbf{T_n,}$
$\mathbf{K45_n,}$
$\mathbf{KD45_n,}$
$\mathbf{S5_n}$ satisfy UIP.

There have also been some studies of UIP in modal logic based on proof-theoretic approaches.
B{\'\i}lkov{\'a}  developed the method originating from Pitts \cite{pitts1992interpretation} to show UIP in  modal logics $\mathbf{K}$ and $\mathbf{KT}$ \cite{bilkova2007uniform}, provability logic  $\mathbf{GL}$ and $\mathbf{Grz}$
\cite{bilkova2016}.
UIP for $\mathbf{K}$, $\mathbf{D}$, $\mathbf{T}$, and $\mathbf{S5}$ via nested sequents and hypersequents has been established by van der Giessen et al. \cite{Giessen2024_uip_proof} 
Furthermore, F{\'e}r{\'e}e et al. \cite{Feree2024} mechanized  the computation of UIP and proved correctness in the
Coq proof assistant for three modal logics:
$\mathbf{K}$, G{\"o}del-L{\"o}b logic $\mathbf{GL}$, intuitionistic strong L{\"o}b logic $\mathbf{iSL}$.
However, these  studies mentioned above did not include  agent symbols in the interpolant formula.


Hintikka 
uses the interpolation theorem to show that
`` logic does
indeed play an important role in explanations and in explaining"\cite[p.161]{Hintikka2007}.
He pointed out that the interpolation theorem is ``by itself a means of explanation"  \cite[p.167]{Hintikka2007}.
In the first-order language,
``the interpolation formula that in a sense serves as `the'
explanation serves as the antecedent of the covering law" \cite[p.178]{Hintikka2007} (which refers to Hempel’s covering law theory of explanation. \cite{Hempel1965}).   
\footnote{The author was indebted to Prof. Katsuhiko Sano for uncovering this point.}
When we connect interpolation with logical explanation as Hintikka suggested, it is very natural to also take agents into consideration in a multi-agent scenario.
Technically, taking agent symbols into the common language of the interpolant formula in the Craig interpolation theorem has been studied in for example \cite{murai20,Su2021_diel_lori}. 
Since UIP is a strengthening of Craig interpolation (however, it is noted that UIP is not established for the first-order language), by
taking agent symbols into consideration, 
it will be helpful when we focus on the influence exerted by a certain group of agents on a logical explanation.
Also, in a different context, UIP is also related to the notion of “forgetting” in knowledge representation and reasoning \cite{Lin1997,hans2009}. The uniform agent interpolation theorem is able to capture the idea that some groups of agents are forgotten in reasoning.


This paper combines the sequent calculi provided in Murai and Sano \cite{murai20}  with the methods studied by B{\'\i}lkov{\'a} \cite{bilkova2007uniform} to show the uniform interpolation for
epistemic logic   $\mathbf{K}_D$, $\mathbf{KD}_D$ and  $\mathbf{KT}_D$ with distributed knowledge. 
Two new terminating sequent calculi are provided.
It provides a 
proof-theoretic proof of UIP for these systems, and presents a purely syntactic algorithm for determining uniform interpolant formulas.
In the definition of an interpolant formula, not only propositional variables but also agent symbols are  taken into consideration.

The contents are organized as follows:
 In section \ref{sec:syntax}, we present our syntax and our Hilbert systems. 
  Section \ref{sec:seq and struc} introduces the main sequent calculi $\gkd$, $\gkdd$, and $\gktd$, and proves their proof-theoretic properties. 
   In section \ref{sec:finite model pro}, it provides an argument for the finite model property  of the proposed sequent calculi, hence it also gives an alternative proof of the admissibility of the cut rule.
  Section \ref{sec:main thm of gkd and gkdd} proves the main theorem of uniform interpolation properties for $\gkd$ and $\gkdd$.
 Section \ref{sec:T} presents a sequent calculus  $\gktdplus$ with a loop-preventing mechanism and shows its structural properties.
 Then, the section \ref{sec:main thm of T} proves main theorem of uniform interpolation properties for $\gktd$ via  $\gktdplus$. Finally, the uniform interpolation properties of both pre-interpolant and post-interpolant formulas with multiple
agent symbols and propositional variables are obtained for these systems.

\section{Syntax}
\label{sec:syntax}

We fix a finite set $\mathsf{Agt}$ of agents, set of  a countable set $\mathsf{Prop}$ of propositional variables.
The set of all  non-empty subsets of $\mathsf{Agt}$ is $\mathsf{Grp}$.
The set of formulas of the language $\mathcal{L}$ is defined inductively as:

\begin{center}
$\alpha ::= p   \,| \bot \,| \alpha\wedge \alpha \,|  \alpha\lor \alpha\,| \alpha\rightarrow \alpha| \neg \alpha \,| D_G \alpha,
\footnote{
 In B{\'\i}lkov{\'a} \cite{bilkova2007uniform},
propositional constants ( $\bot$ or $\top$) are not primitive in the syntax. Generally, we can define $\bot$ as for example $p\wedge \neg p$. If we do so, we need to carefully remove the undesired propositional variables  occurring in the interpolant formulas, to avoid a potential violation of the sharing condition.
However,
as was mentioned in Ono \cite{Ono1998}, $\Box \bot\rightarrow \bot$ is not provable in ${\bf K}$,  it is not always possible to remove all propositional constants. 
To avoid such a difficulty, $\bot$ is taken as  primitive in this syntax.}
 $\end{center}
such that
$p\in\mathsf{Prop} $ and $G\in \mathsf{Grp}.$  
Greek alphabet in uppercase letters, for example, $\Gamma,\Delta$, 
will be used to represent multi-sets of formulas. 
The formula $\langle D_{G}\rangle \alpha$ is $\neg  D_{G} \neg  \alpha$.
Let $n\in\mathbb{N}$ and $0\leq n$, $\overrightarrow{\alpha_n}$ stands for  formulas $\alpha_1,\cdots, \alpha_n$, 
$\overrightarrow{D_{G_n}\alpha_n}$ stands for  formulas $D_{G_1}\alpha_1,\cdots, D_{G_n}\alpha_n$. 
We use $\mathsf{V}(\alpha)$ and $\mathsf{Agt}(\alpha)$ to denote the set of all propositional variables and agent symbols in a formula $\alpha$.
 Similarly, given a multiset $\Gamma$ of formulas, $\mathsf{V}(\Gamma)$ and $\mathsf{Agt}(\Gamma)$ 
  denote the set of all propositional variables and agent symbols  in a formula in $\Gamma$. 
  

\begin{defn}
 A formula in the form of $D_G \alpha$ is called an {\it outmost $G$-boxed formula}, or simply  {\it outmost boxed formula}.
Let $H,G\in \mathsf{Grp}$, a formula  in the form of $D_H \beta$ for $H\subseteq G$
   is called a {\it $G$-sub formula}.
Furthermore, given a finite multiset $\Gamma$ of formulas, let $G,H\in \mathsf{Grp}$, $a\in \mathsf{Agt}$, we define:
$D_G \Gamma= \{D_G \alpha | \alpha\in \Gamma \}$,
 $\Gamma^{\natural_G}= \{ D_G \alpha |  D_G \alpha\in \Gamma \},$ 
 $\Gamma^{\flat_G} = \{    \alpha |  D_G \alpha\in \Gamma  \} $,
$\Gamma^{\natural_{\ni a}}= \{ D_G \alpha |  a\in G ~\text{and}~D_G \alpha\in \Gamma \}$, 
$\Gamma^{\flat_{\subseteq G}} = \{   \alpha |  D_H\alpha\in \Gamma ~\text{for some}~H\subseteq G  \} $.

\end{defn}





\begin{defn}
We define the set $\mathsf{Sub}(\alpha)$ of all subformulas of the formula $\alpha$ inductively as follows:
\begin{scriptsize}
\[
\begin{array}{rcl}
\mathsf{Sub}(p) & := & \{p\}, \\
\mathsf{Sub}(\bot) & := & \{\bot\}, \\
\mathsf{Sub}(D_G\alpha) & := & \mathsf{Sub}(\alpha) \cup \{D_G\alpha\}, \\
 \mathsf{Sub}(\neg\alpha) & := & \mathsf{Sub}(\alpha) \cup \{\neg\alpha\}, \\
\mathsf{Sub}(\alpha \circ \beta) & := & \mathsf{Sub}(\alpha) \cup \mathsf{Sub}(\beta) \cup \{\alpha \circ \beta\},
\end{array}
\]
\end{scriptsize}
where $\circ \in \{\land, \lor, \rightarrow\}$. Given a set $\Gamma$ of formulas,
$\mathsf{Sub}(\Gamma) := \bigcup \{ \mathsf{Sub}(\alpha) \mid \alpha \in \Gamma \}.$

\end{defn}

\begin{defn}
    The {\sl weight} of a $\mathcal{L}$-formula $\alpha$, noted as $\mathsf{wt}(\alpha)$ is inductively defined as: 
    \begin{scriptsize}
    \begin{center}
         $\mathsf{wt}(p)=\mathsf{wt}(\bot)=1$
         
    $\mathsf{wt}(\neg \alpha)=\mathsf{wt}(D_G \alpha)=\mathsf{wt}(\alpha)+1$
    
     $\mathsf{wt}(\alpha\wedge \beta)=\mathsf{wt}(\alpha\lor \beta)= \mathsf{wt}(\alpha\rightarrow \beta)=\mathsf{wt}(\alpha)+\mathsf{wt}(\beta)+1$
    \end{center}
\end{scriptsize}

    Given a multiset $\Gamma$ of formula, $\mathsf{wt}(\Gamma)$ denotes the sum of all $\mathsf{wt}(\alpha)$ for $\alpha\in \Gamma$.
   
\end{defn}





 
  






\begin{table}[htb]
\caption{Hilbert Systems}
\label{table:hilbert}
\begin{footnotesize}
\begin{tabular}{| l l l l |} 
\hline
\multicolumn{4} {|c|} { Hilbert system $\hkd$ } \\
\hline
 & All instantiations of propositional tautologies & &\\
 ($D_K$) & $D_G (\alpha \naraba \beta) \rightarrow (D_G \alpha \rightarrow D_G \beta)$   &
(Incl) & $D_G \alpha \rightarrow  D_H \alpha ~~ (G\subseteq H)$  \\

 (MP) &From $\alpha$ and $\alpha\rightarrow \beta$, infer $\beta$. &
 (Nec) & From $\alpha$, infer $D_G \alpha$. \\

\hline

\multicolumn{4}{| c |}{ Hilbert system $\hkdd$  } \\
\hline
\multicolumn{4}{|l| }{In addition to all axioms and rules of $\hkd$, we add  } \\
($D_D$) & $\neg D_{\{a\}} \bot$ &&\\
\hline
\multicolumn{4}{| c |}{Hilbert system $\hktd$  } \\
\hline
\multicolumn{4}{|l| }{In addition to all axioms and rules of $\hkd$, we add  } \\
($D_T$) & $ D_G  \alpha \rightarrow \alpha$&&\\
\hline
\end{tabular}
\end{footnotesize}
\end{table}
\noindent







\begin{defn}
     The Hilbert systems of $\hkd$, $\hkdd$, $\hktd$   are defined in Table \ref{table:hilbert}.\footnote{In the system of $\hkdd$, the axiom $(D_D)$ is restricted to a single agent, since the seriality for a single agent's binary relation is not always preserved under the operation of intersection (cf. \cite{agotnes20}).}
     Let $\mathbf{L}\in \{    
        \hkd, \hkdd, \hktd
     \}$. Given a set $\Gamma \cup \setof{\alpha}$ of formulas, when we write $\Gamma \vdash_{\mathbf{L}} \alpha$,
we  mean that $\alpha$ is derivable from $\Gamma$ in $\mathbf{L}$  (if the underlying Hilbert system is clear from the context, we simply write $\Gamma \vdash \alpha$). In particular, when $\Gamma$ is empty, we simply write $\vdash \alpha$ instead of $\emptyset \vdash \alpha$. 

\end{defn}

\section{Sequent Calculi and Structural Properties}
\label{sec:seq and struc}



Next, let us move to Gentzen system.
A  {\it sequent}, denoted by $\Gamma \ra \Delta$, is a pair of finite multisets of formulas. 
The multiset $\Gamma$ is the {\em antecedent} of $\Gamma \ra \Delta$,  
while $\Delta$ is the {\em succedent} of the sequent $\Gamma \ra \Delta$. 
A sequent $\Gamma\ra \Delta$  can be read as ``if all formulas in $\Gamma$ hold then some formulas in $\Delta$ hold.''

The logical rules and initial sequents in the following sequent calculi are the same with those in a system named as {\bf G3cp} \cite[p.49]{Negri2001}.
Modal rules here are developed on the basis of the rules in Murai and Sano \cite{murai20} (which established the Gentzen style rules for distributed knowledge) and Su \cite{su2026_awpl_forth} (which established the rules for UIP in multi-agent modal logic \footnote{For this one, those rules are based on the single modal rules from \cite{bilkova2007uniform} and \cite{hakli2012does}.}).
It is noted that, in the rule $(D_K)$, $n$ can be $0$.

\begin{table}[htb]
\caption{Sequent Calculi of $\gkd,\gkdd,\gktd$.}
\label{table:sequent calculi k,kd}

\begin{footnotesize}

\hrule
\begin{tabular}{ll} 
\multicolumn{2}{c}{ Sequent Calculus $\gkd$: } \\

{ \bf Initial Sequents}   & $\Gamma ,p \ra p, \Delta$\hspace{15pt} $\bot, \Gamma\ra \Delta$  \\ 
 \\







{\bf Logical Rules} & 

\AxiomC{$\Gamma \ra \Delta ,\alpha_1$}
\AxiomC{$\Gamma \ra \Delta ,\alpha_2$}
\RightLabel{\scriptsize $(R\wedge )$}
\BinaryInfC{$\Gamma \ra  \Delta ,\alpha_1\wedge \alpha_2$}
\DisplayProof

\AxiomC{$\alpha_1, \alpha_2,\Gamma \ra \Delta$}
\RightLabel{\scriptsize $(L\wedge )$}
\UnaryInfC{$\alpha_1\wedge \alpha_2,\Gamma \ra \Delta$}
\DisplayProof
\\

\,  &

\AxiomC{$\Gamma \ra \Delta,\alpha_1, \alpha_2$}
\RightLabel{\scriptsize $(R\lor )$}
\UnaryInfC{$\Gamma \ra \Delta,\alpha_1\lor \alpha_2$}
\DisplayProof

\AxiomC{$\alpha_1,\Gamma \ra \Delta$}
\AxiomC{$\alpha_2,\Gamma \ra \Delta$}
\RightLabel{\scriptsize $(L\lor )$}
\BinaryInfC{$\alpha_1\lor \alpha_2, \Gamma \ra \Delta$}
\DisplayProof

\\

\,  &

\AxiomC{$\alpha_1,\Gamma \ra \Delta, \alpha_2$}
\RightLabel{\scriptsize $(R\rightarrow )$}
\UnaryInfC{$\Gamma \ra \Delta,\alpha_1\rightarrow \alpha_2$}
\DisplayProof

\AxiomC{$\Gamma \ra \Delta, \alpha_1$}
\AxiomC{$\alpha_2,\Gamma \ra \Delta$}
\RightLabel{\scriptsize $(L\rightarrow )$}
\BinaryInfC{$\alpha_1\rightarrow \alpha_2, \Gamma \ra \Delta$}
\DisplayProof

\\

\,  &
\AxiomC{$\alpha, \Gamma \ra  \Delta$}
\RightLabel{\scriptsize $(R\neg )$}
\UnaryInfC{$\Gamma \ra \Delta,\neg \alpha$}
\DisplayProof

\AxiomC{$\Gamma \ra \Delta, \alpha$}
\RightLabel{\scriptsize $(L\neg )$}
\UnaryInfC{$\neg \alpha, \Gamma \ra  \Delta$}
\DisplayProof
\\ 
\\



{\bf Modal Rule} &

\AxiomC{$\alpha_1,\ldots,\alpha_n \ra \beta  $ }
\RightLabel{{\scriptsize $(D_{K })$}$\dagger$  for any $i$ $(0\leq i \leq n)$, $G_i\subseteq G$  }
\UnaryInfC{$\Sigma,   D_{G_1}\alpha_1,\ldots,D_{G_n}\alpha_n   \ra D_G \beta, \Omega$ }
\DisplayProof
\\



\multicolumn{2}{l}{ {\footnotesize $\dagger$: $\Sigma$ contains only propositional variables, $\bot$ or outmost-boxed formulas which are not $G$-sub formulas.  $\Omega$ contains only  }
   } \\
   
\multicolumn{2}{l}{ {\footnotesize  propositional variables, $\bot$ or any outmost-boxed formulas. }
 }\\

    
\hline

\multicolumn{2}{c}{ Sequent Calculus $\gkdd$
 } \\
 
\multicolumn{2}{c}{ Adding the following rules  to $\gkd$ 
 } \\

{\bf Modal Rule}  &

\AxiomC{$ \Gamma \ra $ }
\RightLabel{\scriptsize $(D_{D})$$\ddagger$}
\UnaryInfC{$\Sigma, D_{\{a\}}  \Gamma   \ra \Omega$}
\DisplayProof
\\


\multicolumn{2}{l}{ {\footnotesize $\ddagger$: $\Sigma$ contains only propositional variables, $\bot$ or  outmost-boxed formulas which are not $\{a\}$-boxed.  $\Omega$ contains only  }
   } \\
\multicolumn{2}{l}{ {\footnotesize  propositional variables, $\bot$ or outmost-boxed formulas. Also, $\Gamma \neq \emptyset$.}
 }\\

\hline
\multicolumn{2}{c}{ Sequent Calculus $\gktd$
 } \\
\multicolumn{2}{c}{ Adding the following rules  to $\gkd$ 
 } \\

{\bf Modal Rule}  &

\AxiomC{$ D_G \alpha, \alpha, \Gamma \ra \Delta $ }
\RightLabel{\scriptsize $(D_{T})$}
\UnaryInfC{$D_G \alpha,\Gamma   \ra \Delta$}
\DisplayProof

\\

\hline

\end{tabular}
\end{footnotesize}
\end{table}

\begin{defn}
    In $\gkd$, $\gkdd$ and $\gktd$, we say that the formulas (or multisets) not in $\Gamma$ and $\Delta$ are 
     {\it principal} in all rules except $(D_{K})$ and $(D_{D})$ . In the rule of  $(D_{K})$, the formulas 
     (or multisets) not in $\Sigma,\Omega$ are principal. In the rule of $(D_{D})$, the formulas in $\Gamma$ are defined
      as principal formulas. 
    A formula (or multiset) is called {\it context} in a rule if it is not principal. 
\end{defn}

In the rules of $(D_K)$ and $(D_D)$, the multisets $\Sigma$ and $\Omega$ denote the context in the derivation.
For a standard $\mathbf{G3}$-style sequent calculus, for example in \cite{hakli2012does}, we generally let the context 
contain arbitrary formulas. Here, we restrict $\Sigma$ and $\Omega$ to contain only propositional variables, constants and modal formulas to make sure that this rule can be well-fitted with the intended step in the definition of the interpolant formula. The price to pay is that the admissibility of weakening is no longer height-preserving.
 Also, we let $\Sigma$ contain modal formulas that  are totally distinguished from principal formulas, in order to, 
in the one hand, make the syntactic definition sort out the principal formulas correctly in the case (iii) 
of the main theorem \ref{thm:main theorem of gkn};
in the other hand, make weakening rules work well for an arbitrary modal formula.

\begin{defn}
\label{def:g1_deri}

Let $\mathbf{L}\in \{ \mathbf{K}_D , \mathbf{KD}_D, \mathbf{KT}_D \}$.
A derivation ${\mathcal D}$ in $ \mathsf{G}(\mathbf{L})$ is a finite tree generated by the rules of $ \mathsf{G}(\mathbf{L})$ from the initial sequents of $ \mathsf{G}(\mathbf{L})$. We say that the {\it end sequent} of ${\mathcal D}$ is the sequent in the root node of ${\mathcal D}$. The {\it height} $n$ of a derivation is the maximum length of the branches in the derivation from the end sequent to an initial sequent.
A sequent $\Gamma \ra \Delta$ is {\it derivable} in $ \mathsf{G}(\mathbf{L})$ (notation: $ \mathsf{G}(\mathbf{L}) \vdash \Gamma \ra \Delta$) if it has a derivation ${\mathcal D}$ in  $ \mathsf{G}(\mathbf{L})$ whose end sequent is $\Gamma \ra \Delta$. Notation  $ \mathsf{G}(\mathbf{L}) \vdash_{n} \Gamma \ra \Delta$ stands for the height of the derivation of that sequent.


\end{defn}

\begin{defn} 
\label{def:g1_context_principal}
Let $\mathbf{L}\in \{ \mathbf{K}_D , \mathbf{KD}_D, \mathbf{KT}_D \}$  and $ \mathsf{G}(\mathbf{L})$ be one of
the systems of Table \ref{table:sequent calculi k,kd}. 
The sequent calculus $ \mathsf{G}^{c}(\mathbf{L})$  is obtained from adding the following cut rule into 
  $ \mathsf{G}(\mathbf{L})$:

\begin{center}
  \begin{scriptsize}
 \AxiomC{$\Gamma \ra  \Delta, \lambda$}
\AxiomC{$\lambda,\Gamma ^{\prime} \ra \Delta ^{\prime} $}
\RightLabel{\scriptsize $(Cut)$}
\BinaryInfC{$\Gamma,\Gamma^{\prime} \ra \Delta, \Delta  ^{\prime}. $}
\DisplayProof
  \end{scriptsize}
\end{center}

\end{defn}

When we consider the derivability of a sequent,
we will try to construct a proof-tree from it upward.
If we are able to find a tree with every branch is in the form of the initial sequents, we say that this sequent is provable.
This procedure is called,  the {\it backward
proof-search} algorithm. (Check Ono\cite{Ono1998} for more discussions.)



\begin{defn}
We define a well-ordered relation of sequent:
$(\Gamma\ra \Delta)\prec (\Gamma^\prime\ra \Delta^\prime) $ if and only if 
$\mathsf{wt}(\Gamma,\Delta)   <
\mathsf{wt}(\Gamma^\prime,\Delta^\prime)$     
\end{defn}


By observing the weight of premises and conclusions of rules, we have:

\begin{prop}
\label{prop:termination of proof search in GKn}
   An arbitrary  backward proof-search in $\gkd$ and $\gkdd$ always terminates.
\end{prop}

In the next part, we will see that a backward proof-search in $\gktd $ does not always terminate. 
This fact requires us to provide
another sequent for $\gktd$ with loop-preventing mechanism.

\begin{defn}
    We say a rule is {\it admissible}, if for an instance of the rule, all premises are derivable, then there is derivation of its conclusion. 
    We say a rule is {\it height-preserving admissible} if for an instance of the rule, all premises are derivable with the greatest height $n$, then there is derivation of its conclusion with the height not greater than $n$. 
   A rule is   
    {\it height-preserving invertible} if for an instance of the rule, if the conclusion has a derivation with the height $n$, then each premise has a derivation with height not greater than $n$. Similarly, we say a rule is   
    {\it invertible} if we drop the above condition of height.
    
\end{defn}


For the details of the following proof of structural properties, please check \cite{Negri2001,troelstra2000basic,Kashima2009}.

\begin{prop}
\label{prop:weakening of k,d,t}
Let $\mathbf{L}\in \{ \mathbf{K}_D , \mathbf{KD}_D, \mathbf{KT}_D \}$. The weakening rules are admissible in $\mathsf{G}(\mathbf{L})$.

\begin{center}
    
 \scalebox{0.8}{
\AxiomC{$\Gamma \ra \Delta$}
\RightLabel{\scriptsize $(RW)$}
\UnaryInfC{$\Gamma \ra \Delta , \lambda$}
\DisplayProof
\AxiomC{$\Gamma \ra \Delta$}
\RightLabel{\scriptsize $(LW)$}
\UnaryInfC{$\lambda,\Gamma \ra \Delta$}
\DisplayProof}
\end{center}
\end{prop}
\begin{proof}
    We proceed by double induction on weight of the formula  $\lambda$ and height of the derivation. 
    It is noted that the weakening rules are not height-preserving.\qedhere

\end{proof}

\begin{prop}
\label{prop:invertibility of all logical rules in Gkn}

       All logical rules in $\gkd$  except $(D_{K})$ are height-preserving invertible. 
       All logical rules in $\gkdd$ except $(D_{K})$ and $(D_{D})$ are height-preserving invertible. 
          All logical rules in $\gktd$ except $(D_{K})$ and $(D_{T})$ are height-preserving invertible. 

\end{prop}


\begin{prop}
    \label{prop:hp contraction in k,d,t}
    Let $\mathbf{L}\in \{ \mathbf{K}_D , \mathbf{KD}_D, \mathbf{KT}_D \}$  
The contraction rules are height-preserving admissible in $\mathsf{G}(\mathbf{L})$.
    
  \begin{center} 
    \scalebox{0.8}{
 \AxiomC{$\Gamma \ra \Delta , \lambda,\lambda$}
\RightLabel{\scriptsize $(RC)$}
\UnaryInfC{$\Gamma \ra \Delta ,\lambda $}
\DisplayProof
\AxiomC{$\lambda,\lambda,\Gamma \ra \Delta$}
\RightLabel{\scriptsize $(LC)$}
\UnaryInfC{$\lambda,\Gamma \ra \Delta$}
\DisplayProof}
    \end{center}
\end{prop}

\begin{proof}
    The proof is done  simultaneously by induction on the weight of contraction formulas and the height of derivations of the premises.
        When the contraction formula $\lambda$ is principal in the end of the derivation,
        Proposition \ref{prop:invertibility of all logical rules in Gkn} is needed for the cases of logical rules.
    When the contraction formula $\lambda$ is in the form of $D_H \alpha$ and the last rule in the derivation is 
      $(D_{K})$, we first consider the contraction formula is not principal:
 \begin{center}
    \scalebox{0.8}{
\AxiomC{$\gamma_1,\ldots, \gamma_m\ra \beta$}
\RightLabel{\scriptsize $(D_{K})$}
\UnaryInfC{$\Sigma,  D_H\alpha,D_H\alpha,    D_{G_1}\gamma_1,\ldots,D_{G_m}\gamma_m   \ra D_G \beta, \Omega$}
\DisplayProof
$~\rightsquigarrow~$
\AxiomC{$\gamma_1,\ldots, \gamma_m\ra \beta$}
\RightLabel{\scriptsize $(D_{K})$}
\UnaryInfC{$\Sigma,  D_H\alpha,     D_{G_1}\gamma_1,\ldots,D_{G_m}\gamma_m   \ra D_G \beta, \Omega$}
\DisplayProof }
\end{center}
where $ G_1 \cup\cdots\cup G_m\subseteq G$ and $H\nsubseteq G$. Next, we consider the case in which $\lambda$ is principal:

  \begin{center}
    \scalebox{0.8}{
\AxiomC{$\alpha,\alpha,\gamma_1,\ldots, \gamma_m\ra \beta$}
\RightLabel{\scriptsize $(D_{K})$}
\UnaryInfC{$\Sigma,  D_H\alpha,D_H\alpha,   D_{G_1}\gamma_1,\ldots,D_{G_m}\gamma_m   \ra D_G \beta, \Omega$}
\DisplayProof
$~\rightsquigarrow~$
\AxiomC{I.H.}
\noLine
\UnaryInfC{$\alpha, \gamma_1,\ldots, \gamma_m\ra \beta$}
\RightLabel{\scriptsize $(D_{K})$}
\UnaryInfC{$\Sigma,  D_H\alpha,    D_{G_1}\gamma_1,\ldots,D_{G_m}\gamma_m   \ra D_G \beta, \Omega$}
\DisplayProof }
\end{center}
where $H\cup G_1 \cup\cdots\cup G_m\subseteq G.$
In the right part, we apply the induction hypothesis to the premise of the assumption, 
since the weight of contraction formula is decreased.
It is noted that the repetition of a boxed formula in the premise of $(D_{T})$ is needed to prove the height-preserving admissibility of
contraction rules
\cite[Chapter 9.1]{troelstra2000basic}.
\end{proof}

Next, the admissibility of cut rule is shown.
Since the weakening rules do not satisfy the height-preserving admissibility, 
we cannot directly apply the method for a standard G3-style sequent calculi for example in \cite[Theorem 4.15]{troelstra2000basic}. 
Here, the proof is done by following a similar argument in \cite[Theorem 3.23]{Negri2001}. Also, it is noted that the cut rule is in the form without shared context. 

\begin{prop}
\label{prop:cut admissible of k,d,t}
 
Let $\mathbf{L}\in \{ \mathbf{K}_D , \mathbf{KD}_D, \mathbf{KT}_D \}$    The cut rule is admissible in $\mathsf{G}(\mathbf{L})$.
    


    
\end{prop}


\begin{center}

  \begin{scriptsize}
\begin{prooftree}
\AxiomC{$\vdots~\mathcal {D}_1$}
\RightLabel{$\mathtt{rule} (\mathcal{D}_1)$} 
\UnaryInfC{$\Gamma \ra  \Delta ,\lambda$}
\AxiomC{$\vdots~\mathcal {D}_2$}
\RightLabel{$\mathtt{rule} (\mathcal{D}_2)$} 
\UnaryInfC{$\lambda , \Gamma ^{\prime}  \ra \Delta^{\prime} $}
\RightLabel{$(cut)$} 
\BinaryInfC{$\Gamma,\Gamma ^{\prime}   \ra \Delta, \Delta^{\prime}$}
\end{prooftree}
  \end{scriptsize}
\end{center}

\begin{proof}
It is shown that if a $(cut)$  appears only in the end of a derivation $\mathcal{D}$, then there is a derivation in which no $(cut)$ appears and  ends with the same conclusion as $\mathcal{D}$. 
This can be proved by double induction on the {\it complexity} (the number of logical connectives of the cut formulas of $(cut)$) and the {\it height}, i.e., the number of all the sequents in the derivation. 
The argument is divided into the following three cases: 
\begin{itemize}
\item[(1)] ${\mathcal D_1}$ or $ \mathcal{D}_2$ is an initial sequent.  
\item[(2)] $\mathtt{rule} (\mathcal{D}_1)$ or $\mathtt{rule} (\mathcal{D}_2)$ is a logical or modal rule in which the cut formula is not principal. 
\item[(3)] $\mathtt{rule} (\mathcal{D}_1)$ and $\mathtt{rule} (\mathcal{D}_2)$ are logical or modal rules, and the cut formulas are principal in both rules.
\end{itemize}

The proof proceeds basically  following an ordinary argument for the admissibility of cut 
for classical propositional logic \cite[Theorem 3.23]{Negri2001}.

We only show some examples in the case (3). At first, we consider the case in which $\mathtt{rule} (\mathcal{D}_1)$ is $(D_{K})$, 
$\mathtt{rule} (\mathcal{D}_2)$ is $(D_{D})$ and cut formulas are principal in both rules.

\begin{center}
      \begin{scriptsize}
\noLine
\AxiomC{$\mathcal {D}_1$}
\UnaryInfC{$\Gamma\ra \lambda$}
\RightLabel{\scriptsize $(D_{K})$}
\UnaryInfC{$\Sigma,D_{\{a\}}\Gamma \ra D_{\{a\}}\lambda,\Omega$}
\noLine
\AxiomC{$\mathcal {D}_2$}
\UnaryInfC{$\lambda ,\Gamma' \ra   $}
\RightLabel{\scriptsize $(D_{D})$}
\UnaryInfC{$D_{\{a\}}\lambda,\Sigma^{\prime},   D_{\{a\}} \Gamma^{\prime}   \ra   \Omega^{\prime}$}
\RightLabel{\scriptsize $(cut)$}
\BinaryInfC{$\Sigma, \Sigma^{\prime},  D_{\{a\}}\Gamma , D_{\{a\}}\Gamma^{\prime}  \ra \Omega,\Omega^{\prime}$}
\DisplayProof
      \end{scriptsize}
\end{center}

Then we can transform the derivation into the following:

\begin{center}
      \begin{scriptsize}
\noLine
\AxiomC{$\mathcal {D}_1$}
\UnaryInfC{$\Gamma \ra \lambda$} 
\noLine
\AxiomC{$\mathcal {D}_2$}
\UnaryInfC{$\lambda,\Gamma^{\prime} \ra  $} 
\RightLabel{\scriptsize $(cut)$}
\BinaryInfC{$\Gamma,\Gamma^{\prime} \ra$}
\RightLabel{\scriptsize $(D_{D})$}
\UnaryInfC{$\Sigma, \Sigma^{\prime},  D_{\{a\}}\Gamma , D_{\{a\}} \Gamma^{\prime}  \ra \Omega,\Omega^{\prime}$}
\DisplayProof
      \end{scriptsize}
\end{center}
In the transformed derivation, the application of $(cut)$ can be eliminated owing to the lower complexity of the cut formula.

In the case (3), when $\mathtt{rule} (\mathcal{D}_1)$ and $\mathtt{rule} (\mathcal{D}_2)$ are both   $(D_{K})$, also cut formulas are principal in both rules, the proof is similar to the above case.
If $\mathtt{rule} (\mathcal{D}_1)$ is $(D_{K})$, $\mathtt{rule} (\mathcal{D}_2)$ is $(D_{T})$ and cut formulas are principal in both rules.


\begin{center}
      \begin{scriptsize}
\noLine
\AxiomC{$\mathcal {D}_1$}
\UnaryInfC{$\gamma_1,\ldots,\gamma_m \ra \lambda$}
\RightLabel{\scriptsize $(D_{K})$}
\UnaryInfC{$\Sigma,D_{G_1}\gamma_1,\ldots,D_{G_m}\gamma_m \ra D_{G}\lambda,\Omega$}
\noLine
\AxiomC{$\mathcal {D}_2$}
\UnaryInfC{$D_G \lambda, \lambda,\Gamma^{\prime} \ra \Delta $}
\RightLabel{\scriptsize $(D_{T})$}
\UnaryInfC{$D_G \lambda,  \Gamma^{\prime} \ra \Delta $}
\RightLabel{\scriptsize $(cut)$}
\BinaryInfC{$\Sigma,   D_{G_1}\gamma_1,\ldots,D_{G_m}\gamma_m ,  \Gamma^{\prime}  \ra \Delta,\Omega $}
\DisplayProof
      \end{scriptsize}
\end{center}
We can transform the derivation into the following:
\begin{center}
      \begin{scriptsize}
\noLine
\AxiomC{$\mathcal {D}_1$}
\UnaryInfC{$\gamma_1,\ldots,\gamma_m \ra \lambda$}
\noLine
\AxiomC{$\mathcal {D}_1$}
\UnaryInfC{$\Sigma,D_{G_1}\gamma_1,\ldots,D_{G_m}\gamma_m \ra D_{G}\lambda,\Omega$}
\noLine
\AxiomC{$\mathcal {D}_2$}
\UnaryInfC{$D_G \lambda, \lambda,\Gamma^{\prime} \ra \Delta $}
\RightLabel{\scriptsize $(cut)$}
\BinaryInfC{$\Sigma,D_{G_1}\gamma_1,\ldots,D_{G_m}\gamma_m, \lambda,\Gamma^{\prime} \ra \Omega,\Delta $}
\RightLabel{\scriptsize $(cut)$}
\BinaryInfC{$\gamma_1,\ldots,\gamma_m,\Sigma,D_{G_1}\gamma_1,\ldots,D_{G_m}\gamma_m, \Gamma^{\prime} \ra\Omega,\Delta $}
\RightLabel{\scriptsize $(D_{T})^\ast$}
\UnaryInfC{$\Sigma,D_{G_1}\gamma_1,\ldots,D_{G_m}\gamma_m, \Gamma^{\prime} \ra \Omega,\Delta$}
\DisplayProof
      \end{scriptsize}
\end{center}
where  $(D_{T})^\ast$ denotes applying  $(D_{T})^\ast$ for finite many times.
In the transformed derivation, the uppermost application of $(cut)$ can be eliminated due to the reduced height of the derivation, while the second uppermost application of $(cut)$ can be eliminated owing to the lower complexity of the cut formula.
\end{proof}

Next we show the equipollence between Hilbert systems and sequent calculi.

\begin{defn}
Given a sequent $\Gamma \ra \Delta$, $\Gamma_{\star}$ denotes the conjunction of all formulas in $\Gamma$ $(\Gamma_{\star}\equiv \top$ if $\Gamma$ is empty), $\Delta^{\star}$ denotes the unique formula in $\Delta$  ($\Delta^\star = \bot$ if $\Delta$ is empty).
\end{defn}

\begin{prop}
\label{prop:from g to h}
Let $\mathbf{L}\in \{ \mathbf{K}_D , \mathbf{KD}_D, \mathbf{KT}_D \}$ .
 $\mathsf{G}(\mathbf{L})   \vdash  \Gamma \ra \Delta$ implies $\mathsf{H}(\mathbf{L})\vdash   \Gamma_{\star} \naraba \Delta ^\star$.
\end{prop}

\begin{thm}[Equipollence]
\label{prop:equivilence of Hilbert and sequent}
Let $\mathbf{L}\in \{ \mathbf{K}_D , \mathbf{KD}_D, \mathbf{KT}_D \}$.
The following equivalence holds: $\mathsf{H}(\mathbf{L}) \vdash        \alpha$ iff 
 $ \mathsf{G}(\mathbf{L})  \vdash\ra \alpha$. 
\end{thm}
\begin{proof}
The direction from the right to the left  can be proved by applying Proposition~\ref{prop:from g to h}, in which we let the antecedent $\Gamma$ be empty. 
The direction from the left to the right  can be proved by induction on the derivation of $\alpha$. 
\end{proof}

\section{Finite model property}
\label{sec:finite model pro}

Next, we turn to establish the finite model properties for the sequent calculi and also provide a semantic proof of  admissibility of the cut rule.
Let $W$ be a set of states, $ (R_G)_{G\in \mathsf{Grp}}$ be
a familiy of binary relation such that  $R_H \subseteq R_G $  if $G\subseteq H$,
$V$ be a valuation
function from  $\mathsf{Prop}$ to $\mathcal{P}(w)$. Then, 
we call a pair $F=(W, (R_G)_{G\in \mathsf{Grp}})$ a frame and
and $M=(W, (R_G)_{G\in \mathsf{Grp}}, V)$ a model. 
Given a $M$ and a state $w$ in $W$, the satisfying relation  is defined inductively.
The cases of propositional cases are defined as usual,
the modal formula is defined as :
    $M, w \models D_G \alpha \quad \text{iff for all } v \in W, 
    \text{if } (w,v) \in  R_G \text{ then } M, v \models \alpha. $

The class of all models is denoted as  $\md$. 
 $\mdd$ is the class of all models in which 
$R_{\{a\}}$ ($\{a\}\in \mathsf{Grp}$) is serial, that is for any $w\in W$ there
 is a $v\in W$ such that $wR_{\{a\}} v$;
 $\mdt$ is the class of all models in which 
$R_G$ ($G\in \mathsf{Grp}$) is reflexive, that is for any $w\in W$  $wR_G w$. 
Furthermore, $\mathbb{M}^\mathbf{F}$  denotes the class of all finite models.
 We use $|M|$ to denote the underlying set
  (or domain) of a model $M$.

Usually, the truth of a distributed knowledge $D_G \alpha $ in a state $w$   is defined as
the truth of $\alpha$ in all states accessible from $w$  by the intersection of all individual accessibility relations in group $G$.
The models above, in which each accessibility relation is indexed by groups of agents, are usually called {\it pseudo models}.
The difference is that 
a binary relation indexed by $G$ does not necessarily coincide with the intersection of the individual relations of the agents in that group.
In a standard argument (for example in Fagin \cite{fagin96}) about the completeness of the epistemic logic with distributed knowledge, we usually make use of a method called tree-unraveling  to transform  
a model into a well-defined model such that $R_{G}$ is 
 the intersection of $\inset{R_{\setof{a}}}{a \in G}$ by preserving the satisfaction of  formulas (\cite{fagin96,agotnes20}). 
However, for our current purpose, the above model is sufficient.
\footnote{If readers are interested in the definition of tree unraveling,
see Blackburn et.al \cite{blackburn01} for a general method to treat tree-like models. Also, please check Dummett \cite{dummett59} and   Sahlqvist \cite{sahlqvist75} for a development of this method, the argument of $\mathbf{S5}$ in Fagin \cite{fagin96},
and an improvement of this method in \cite{wang20}}.

In the rest  of this section,  
the finite model properties  of sequent calculi $\gkd,\gkdd$ and $\gktd$  and the soundness of these sequent calculi with the (cut) rules 
for the above semantics are shown.
Hence, the decidability and a semantic proof of the admissibility of the cut rules for these systems can be derived.





Given a sequent $\Gamma \Rightarrow \Delta$,
we say that a sequent $\Gamma \Rightarrow \Delta$ is \emph{valid} in a 
class $\mathbb{M}$ of models (denoted by $\mathbb{M} \models \Gamma \Rightarrow \Delta$), 
if $\mathbb{M} \models \Gamma_\star \rightarrow \Delta^\star$.
Recall that  $ \mathsf{G}^{c}(\mathbf{L})$ denotes the  sequent calculus $ \mathsf{G}(\mathbf{L})$
with an additional cut rule.
In the following theorem, 
let  $\mathbb{M}_{ \mathbf{K}_D}$
be $\mathbb{M} $ ,
$\mathbb{M}_{ \mathbf{KD}_D}$
be $\mathbb{M}_{\mathbf{ser}}$ ,
$\mathbb{M}_{ \mathbf{KT}_D}$
be $\mathbb{M}_{\mathbf{ref}}$.

\begin{thm}[Soundness]
    \label{thm:soundness of g plus cut}
Let $\mathbf{L}\in \{ \mathbf{K}_D , \mathbf{KD}_D, \mathbf{KT}_D \}$. Let $\Gamma \Rightarrow \Delta$ be a sequent. If $\mathsf{G}^c(\mathbf{L}) \vdash \Gamma \Rightarrow \Delta$ then $\mathbb{M}_{\mathbf{L}}\models \Gamma \Rightarrow \Delta$.
\end{thm}
\begin{proof}
    These can be shown by induction on the derivation.
    We concentrate on the following rule.

    \begin{itemize}
        \item In the rule of $(D_{K })$, given $0\leq i\leq n$ and, we assume that $\mathbb{M}\models  \alpha_1,\ldots,\alpha_n \ra \beta   $ i.e., 
        $\mathbb{M}\models  (\alpha_1\wedge\ldots\wedge \alpha_n )\rightarrow \beta   $.
         We show that $\mathbb{M}\models  \Sigma,   D_{G_1}\alpha_1,\ldots,D_{G_n}\alpha_n   \ra D_G \beta, \Omega$ 
        that is, $\mathbb{M}\models  (\bigwedge\Sigma \wedge  D_{G_1}\alpha_1\wedge \ldots\wedge D_{G_n}\alpha_n  )
         \rightarrow (D_G \beta \lor \bigvee\Omega)$.
          When $n=0$, the proof is straight. We only discuss the case $n$ is not $0$.
         Fix any $M$, a state $w\in |M|$, suppose $M,w\models   \bigwedge\Sigma \wedge  D_{G_1}\alpha_1\wedge \ldots\wedge
          D_{G_n}\alpha_n  $. We show $M,w\models D_G \beta \lor \bigvee\Omega$ which can be obtained from $M,w\models D_G \beta  $.
          Fix a $v\in|M|$ such that $(w,v) \in  R_G $, and we show $M,v\models \beta $.
        Then, for any $i$  we have $   R_G \subseteq   R_{G_{i}} $
         since $G_{i}\subseteq G$. We have
        $ M, w\models D_{G_i} \alpha_i$ for any $i$. 
       Then, we can derive that 
        $M,v\models \alpha_1\wedge\cdots\wedge \alpha_n$. Finally,  $M,v\models \beta $ can be obtained.
       \qedhere
    \end{itemize}
\end{proof}
From this proof,
it is trivial that the above systems without cut are also sound for the corresponding class of models.
Then, we are ready to show the finite model property of these sequent calculi.
The following method is based on Takano \cite{takano2018semantical}.
 In the rest of this section, even
if $\Gamma,\Delta$ are  sets of formulas, we regard $\Gamma\ra \Delta$   as a sequent.

\begin{defn}
We say that a sequent calculus $\mathsf{G}$ enjoys the {\it finite model property} 
with respect to a class $\mathbb{M}^{\mathbf{F}}$ of finite frame, if
$\mathbb{M}^{\mathbf{F}}\models \Gamma\ra \Delta$ implies $\mathsf{G}\vdash \Gamma\ra \Delta.$
\end{defn}

\begin{defn}
Let $\mathbf{L}\in \{ \mathbf{K}_D , \mathbf{KD}_D,\mathbf{KT}_D\}$.
Let $\Omega $ be a finite subformula-closed set of formulas.
Given a sequent $\Gamma\ra\Delta$, we say that it is $\Omega$-saturated in $\mathsf{G}(\mathbf{L})$ if it satisfies the following conditions:
1, $ \mathsf{G}(\mathbf{L}) \nvdash \Gamma\ra\Delta;$
2,  for any formula $\alpha$ in $\Omega$, if $ \mathsf{G}(\mathbf{L}) \nvdash  \alpha,\Gamma\ra \Delta$ then
 $\alpha \in \Gamma$;  3, for any formula $\beta$ in $\Omega$, if $ \mathsf{G}(\mathbf{L}) \nvdash \Gamma\ra \Delta,\beta$ 
    then $\beta \in \Delta$.
\end{defn}

\begin{prop}
    \label{prop:construction model}
Let $\mathbf{L}\in \{ \mathbf{K}_D , \mathbf{KD}_D,\mathbf{KT}_D\}$.
Let $\Gamma\ra\Delta$ be a sequent such that it is not derivable in $ \mathsf{G}(\mathbf{L})$.
Let $\Omega $ be a finite subformula closed set of formulas and $\Gamma,\Delta\subseteq \Omega.$
Then we can construct an $\Omega$-saturated sequent $\Theta\ra \Pi$  in $ \mathsf{G}(\mathbf{L})$ such that $ \Gamma\subseteq\Theta$ and $\Delta\subseteq\Pi.$

\end{prop}
\begin{proof}

Consider an enumeration of all formulas in $\Omega$ ($\Omega$ is finite) as $\phi_1,\ldots,\phi_n.$
We define $(\Theta_k)_{k\leq n}$ inductively. 
Let $\Theta_0 \ra \Pi_0$ be $\Gamma \ra \Delta$, then we define  $\Theta_{k+1} \ra \Pi_{k+1}$ as follows:


\begin{enumerate}

    \item if $\nvdash \Theta_k\ra \Pi_k,\phi_k$, 
    then $\Theta_{k+1}=\Theta_k$, $\Pi_{k+1}=\Pi_k\cup \{\phi_k\}$;
 
    \item if $\vdash \Theta_k\ra \Pi_k,\phi_k$ and $\nvdash \phi_k,\Theta_k\ra \Pi_k $, then $\Theta_{k+1}=\Theta_k  \cup \{\phi_k\}$, $\Pi_{k+1}=\Pi_k $;
 
    \item otherwise,$\Theta_{k+1}=\Theta_k  $, $\Pi_{k+1}=\Pi_k$.
\end{enumerate}
Finally,
let $\Theta = \Theta_{n+1}$, $\Pi=\Pi_{n+1}$.
\end{proof}

\begin{defn}
Let $\mathbf{L}\in \{ \mathbf{K}_D , \mathbf{KD}_D,\mathbf{KT}_D\}$.
Let $\Omega $ be a finite subformula closed set of formulas.
An $\Omega$-model $M_\Omega =( (W, (R_G)_{G\in \mathsf{Grp}}, V))$ is defined as follows:
\begin{enumerate}
    \item $W=\{ (\Theta,\Pi) | \Theta\cup\Pi\subseteq \Omega$ and $\Theta\ra\Pi$ is $\Omega$-saturated in $\mathsf{G}(\mathbf{L})\}$;
  \item $ (\Theta,\Pi) R_G (\Theta',\Pi')$ if and only if $\Theta ^{\flat G'} \subseteq \Theta'$ for any $G'\subseteq G$ such that $G'\in \mathsf{Grp}$;
    \item $(\Theta,\Pi)\in V(p)  $ if and only if $p\in \Theta$.
\end{enumerate}
It is remarked that, $W$ is finite since $\Omega$ is finite.
\end{defn}

\begin{prop}
   Let $\mathbf{L}\in \{ \mathbf{K}_D , \mathbf{KD}_D, \mathbf{KT}_D \}$. Let $\Omega $ be a finite subformula-closed set of formulas.  An $\Omega $-model in $\mathsf{G}(\mathbf{L})$ is in  $\mathbb{M}^{\mathbf{F}}_{\mathbf{L}}$.

      
\end{prop}

\begin{lem}
    \label{lem:truth lemma}
Let $\Omega $ be a finite subformula-closed set of formulas. For any $(\Theta,\Pi) \in |M_\Omega|$, for any formula $\alpha$,
we have:
\begin{enumerate}
    \item $\alpha\in \Theta$ implies $M_{\Omega} , (\Theta,\Pi) \models  \alpha$;
    \item   $\alpha\in \Pi$ implies $M_{\Omega} , (\Theta,\Pi) \nvDash  \alpha$.
\end{enumerate}
\end{lem}
\begin{proof}
We proceed by induction on the $\alpha$.
When $\alpha$ is $D_G \beta$, the argument goes like follows.
We assume $D_G \beta\in \Theta$ and aim to show $M_{\Omega} , (\Theta,\Pi) \models D_G \beta$. Fix any $(\Theta',\Pi')$ such that $(\Theta,\Pi) R_G (\Theta',\Pi')$. By the definition we have $\Theta ^{\flat G} \subseteq \Theta'$ for $G\subseteq G$,
 then it derives that $\beta \in \Theta'$. Applying induction hypothesis, we obtain $M_{\Omega} , (\Theta',\Pi') \models  \beta$.
For the condition (2), we assume that $D_G \beta\in\Pi$.
Since $\nvdash \Theta\ra\Pi $, we have $\nvdash \Theta^{\flat G_1},\cdots,\Theta^{\flat G_n} \ra\beta $ for a finite family of $(G_i)_{1\leq i\leq n}$ of all subgroups of $G$.
To see this, it assumes for contradiction that $\vdash  \Theta^{\flat G_1},\cdots,\Theta^{\flat G_n} \ra\beta$. After applying $(D_K)$  and weakening rules, we derive $\vdash\Theta\ra\Pi $ which leads to a contradiction.
 Then, 
we can construct a $(\Theta',\Pi')\in W$ such that 
$ \Theta^{\flat G_1}\cup\cdots\cup\Theta^{\flat G_n}\subseteq \Theta'$
and $\{\beta\}\subseteq \Pi'$ according to Proposition \ref{prop:construction model}. 
Then, we have $\Theta^{\flat G_i}\subseteq \Theta'$ for any $G_i\subseteq G$.
From the definition of binary relation, we have $(\Theta,\Pi) R_G (\Theta',\Pi')$. Furthermore, $M_{\Omega} , (\Theta',\Pi') \nvDash  \beta$
from $\beta\in \Pi'$ according to induction hypothesis.
Other cases are standard.
\end{proof}



\begin{thm}
    \label{thm:completeness of g without cut}
Let $\mathbf{L}\in \{ \mathbf{K}_D , \mathbf{KD}_D,\mathbf{KT}_D\}$. Given a sequent $\Gamma\ra\Delta$,
if  $\mathbb{M}^{\mathbf{F}}_{\mathbf{L}} \models \Gamma\ra \Delta$ then $\mathsf{G}(\mathbf{L} )\vdash \Gamma\ra \Delta. $

\end{thm}
\begin{proof}
    At first, we assume that $\mathsf{G}(\mathbf{L} )\nvdash \Gamma\ra \Delta. $ 
    We put  $\Omega  = \mathsf{Sub}(\Gamma\cup\Delta)$. It is obvious that $\Omega $ is 
     finite and subformula-closed. By applying Proposition \ref{prop:construction model}, there exists
     an $\Omega$-saturated sequent $\Theta\ra\Pi$ in $\mathsf{G}(\mathbf{L} )$, such that $\Gamma\subseteq\Theta $ and $\Delta\subseteq \Pi$.
     Then, we apply Lemma \ref{lem:truth lemma} to obtain the result.
\end{proof}

\begin{cor}
Sequent calculi $\gkd,\gkdd$ and $\gktd$ are decidable. 

\end{cor}
 

The following result provides an alternative proof for the same result as Proposition \ref{prop:cut admissible of k,d,t}.
\begin{cor}
The cut rule is admissible in  $\gkd,\gkdd$ and $\gktd$.
\end{cor}
\begin{proof}
    Let $\mathbf{L}\in \{ \mathbf{K}_D , \mathbf{KD}_D,\mathbf{KT}_D\}$.
    Given a sequent $\Gamma\ra\Delta$, such that $\mathsf{G}^{c}(\mathbf{L} )\vdash \Gamma\ra \Delta. $
    According to soundness in Theorem \ref{thm:soundness of g plus cut}, 
    we obtain that $\mathbb{M}_{\mathbf{L}} \models \Gamma\ra \Delta.$ It derives that $\mathbb{M}^{\mathbf{F}}_{\mathbf{L}} \models \Gamma\ra \Delta.$
Then, we obtain $\mathsf{G}(\mathbf{L} )\vdash \Gamma\ra \Delta $ from 
Theorem \ref{thm:completeness of g without cut}.
\end{proof}

\section{Main theorem of $\gkd$ and $\gkdd$}
\label{sec:main thm of gkd and gkdd}

Now, we are ready to show the main theorem of $\gkd$ and $\gkdd$.
The method we apply originates in \cite{pitts1992interpretation} for propositional  intuitionistic logic and was further developed by \cite{bilkova2007uniform} for modal logic.
We first fix a propositional variable $p$ and an agent symbol $a$, then 
define an algorithm to generate the interpolant formula $\mathcal{A}_{(p,a)}(\Gamma;\Delta)$ of a sequent $\Gamma\ra \Delta$.
This syntactic algorithm is defined 
in correspondence with  the proof-search on $\Gamma\ra \Delta$.
The generated interpolant formula $\mathcal{A}_{(p,a)}(\Gamma;\Delta)$ which does not contain $p$ and $a$, still maintains the desired derivability.
The idea is to
define the interpolant formula that reflects the operation to
separate and rewrite 
some parts of the proof-tree which are built on the occurrences of $p$ and $a$.


\begin{defn}
\label{dfn:Ap formula in Gkn}
Let $\mathbf{L}\in \{ \mathbf{K}_D , \mathbf{KD}_D\}$ 
      Let $\Gamma,\Delta$ be finite multi-sets of formulas, $p$ be a propositional variable, $a$ be an agent symbol.
          We say that a sequent $\Gamma \ra \Delta$ is  {\it critical} if $\Gamma$ and $\Delta$ contain only propositional variables, $\bot$ or outmost-boxed formulas.
      An $\mathcal{A}$-formula $\mathcal{A}_{(p,a)}(\Gamma;\Delta)$  is defined inductively as follows.

  \begin{scriptsize}
\[
\begin{array} 
{|c|c|c|c|}\hline & \Gamma\ra\Delta ~\text{is not critical and }~ \Gamma;\Delta\text{ matches} & \mathcal{A}_{(p,a)}(\Gamma;\Delta)\text{ equals} \\
\hline
~~1~~ & \Gamma',\alpha_1\wedge \alpha_2;\Delta & \mathcal{A}_{(p,a)}(\Gamma',\alpha_1,\alpha_2;\Delta) \\
2 & \Gamma;\alpha_1\wedge \alpha_2,\Delta' & \mathcal{A}_{(p,a)}(\Gamma;\alpha_1,\Delta')\land \mathcal{A}_{(p,a)}(\Gamma;\alpha_2,\Delta') \\

3 & \Gamma',\alpha_1\vee \alpha_2;\Delta & \mathcal{A}_{(p,a)}(\Gamma',\alpha_1;\Delta)\land \mathcal{A}_{(p,a)}(\Gamma',\alpha_2;\Delta) \\

~~4~~ & \Gamma;\alpha_1\vee \alpha_2,\Delta' & \mathcal{A}_{(p,a)}(\Gamma';\alpha_1,\alpha_2,\Delta') \\

~~5~~ & \Gamma',\neg \alpha;\Delta & \mathcal{A}_{(p,a)}(\Gamma';\alpha,\Delta) \\

6 & \Gamma;\neg \alpha,\Delta' & \mathcal{A}_{(p,a)}(\Gamma,\alpha;\Delta') \\

7 & \Gamma',\alpha_1\rightarrow \alpha_2;\Delta & \mathcal{A}_{(p,a)}(\Gamma';\Delta,\alpha_1)\land \mathcal{A}_p(\Gamma',\alpha_2;\Delta) \\
8 & \Gamma;\alpha_1\rightarrow \alpha_2,\Delta' & \mathcal{A}_{(p,a)}(\Gamma,\alpha_1; \Delta',\alpha_2) \\

\hline
\end{array}
\]

\[
\begin{array}
{|c|c|c|c|}
\hline 
& \Gamma\ra\Delta ~\text{is critical and }~ \Gamma;\Delta\text{ matches} & \mathcal{A}_{(p,a)}(\Gamma;\Delta)\text{ equals} \\
\hline

~~9~~ & \Gamma',p;\Delta',p & \top \\


~~10~~ &  \Gamma',  D_{\{a\}} \Gamma''  ; \Delta', D_{\{a\}} B ~\text{with}~   \mathsf{G}(\mathbf{L})\vdash  \Gamma'' \ra B&  \top \\


~~10' \dagger~~ &  \Gamma', D_{\{a\}} \Gamma'';\Delta ~\text{with}~ \gkdd\vdash \Gamma''\ra   &  \top \\


11\ddagger& ~~ \Phi, \overrightarrow{D_{G_m}\gamma_m } ; \overrightarrow{D_{H_n}\delta_n }, \Psi ~~  &  \mathsf{X} \\

\hline
\end{array}
\]
\end{scriptsize}

{\footnotesize $\dagger:$ The Line $10'$ is only applied in $\gkdd$.}
 
{\footnotesize 
$\ddagger:$ $\Phi$ and  $\Psi$ are multisets containing  propositional variables or $\bot$.
$\Phi\cup\overrightarrow{D_{G_m}\gamma_m } \cup \overrightarrow{D_{H_n}\delta_n }\cup \Psi$is not empty.}

The formula $\mathsf{X}$ is:
  \begin{scriptsize}
\[
\bigvee_{r\in\Phi/\{p\}}\neg r 
\vee 
\bigvee_{q\in\Psi/\{p\}}q
\vee
\bigvee_{\bot\in\Phi}\neg\bot
\vee
\bigvee_{\bot\in\Psi}\bot
\vee
\bigvee_{D_{G_j}\gamma_j \in \{\overrightarrow{D_{G_m}\gamma_m }\} 
   / \{\overrightarrow{D_{G_m}\gamma_m }\}^{\natural_{\ni a}}  }
\langle D_{G_j} \rangle
\mathcal{A}_{(p,a)} (    \{ \overrightarrow{D_{G_m\gamma_m}} \}^{\flat_{\subseteq G_j}} ;\emptyset) 
\]

\[
  \vee \bigvee_{D_{H_i}\delta_i \in \{\overrightarrow{D_{H_n}\delta_n } \}
  / \{\overrightarrow{D_{H_n}\delta_n } \}^{\natural_{\ni a}} }
  D_{H_i} \mathcal{A}_{(p,a)} (   \{\overrightarrow{D_{G_m}\gamma_m }  \}^{\flat_{\subseteq{H_i}}}; \delta_i   )
\]
\[
  \vee
     \bigvee_{D_{H_i}\delta_i \in  \{\overrightarrow{D_{H_n}\delta_n } \}^{\natural_{\ni a}}
     ~\text{with}~H_i \ne \{a\}}
  D_{H_i/\{a\}} \mathcal{A}_{(p,a)} (   \{\overrightarrow{D_{G_m}\gamma_m }  \}^{\flat_{\subseteq{H_i}}}; \delta_i   )
\]
  \end{scriptsize}

Recall that for any formula $\alpha$, $n\in\mathbb{N}$, $\overrightarrow{\alpha_n}$ stands for  $\alpha_1,\cdots, \alpha_n$, $\overrightarrow{D_{G_n}\alpha_n}$ stands for  formulas $D_{G_1}\alpha_1,\cdots, D_{G_n}\alpha_n$. 
 $\Gamma^{\natural_{\ni a}}= \{ D_G \alpha |  a\in G ~\text{and}~D_G \alpha\in \Gamma \}$, 
  $\Gamma^{\flat_{\subseteq G}} = \{   \alpha |  D_H\alpha\in \Gamma ~\text{for some}~H\subseteq G  \} $.
The formula $\mathcal{A}_{(p,a)}(\Gamma;\Delta)$ is defined
in the following procedure: at first, the lines $1-8$ are repeatedly applied until it reaches a critical sequent 
(the order does not matter, since
all propositional rules are height-preserving invertible by Proposition \ref{prop:invertibility of all logical rules in Gkn}). 
Next, we test whether it   matches the line $9$, $10$ (or $10'$ in $\gkdd$).
If no one is the case, the line $11$ is applied. We repeat the above procedure until $\Gamma;\Delta$ cannot  match any lines in the table, in this case $\mathcal{A}_{(p,a)}(\Gamma;\Delta)$ is defined as $\bot$.
If $X$ does not meet these conditions,  $X$ is $\bot$.
 Especially $\mathcal{A}_{(p,a)}(\emptyset;\emptyset)$ is defined as $\bot$.
\footnote{
The lines $10$ and $10'$  above are necessary to deal with case where
 the derivation occurring in the (iii) of the following main Theorem \ref{thm:main theorem of gkn}  ends with the rules 
$(D_{K})$ or $(D_{D})$ with the formula $D_{\{a\}}$ becoming principal. These cases significantly increase the computation time. However, we need these definitions to transform the formulas to an interpolant formula whose intended derivability does not rely on applying modal rules on the agent $a$, since $a$ like the propositional variable $p$ (that is similarly dealt with in the line 9) should not occur in the interpolant formula.}


We observe that for any ${\mathcal A}$-formulas defined in the right part, its weightiness always decrease when compare to its right part.
We can define a well-order relation of $\mathcal{A}$-formulas as follows:
   $  \mathcal{A}_{(p,a)} (\Gamma;\Delta) \prec  \mathcal{A}_{(p,a)} (\Gamma^\prime;\Delta^\prime )$ if and only if
   $\mathsf{wt}(\Gamma\ra \Delta)  \prec  \mathsf{wt}(\Gamma^\prime\ra \Delta^\prime)$.
Given the fact that all back proof-search in $\gkd$, $\gkdd$ always terminates (in Proposition \ref{prop:termination of proof search in GKn}) and decidability for the Line 11, we can see that 
such a formula can always be determined.

\end{defn}

\begin{ex}
Let $\Gamma$ be $D_{\{1\}} q, D_{\{3\}} p$, $\Delta$ be $D_{\{1,2\}} r$.

\[
\begin{array}
{lll}
\mathcal{A}_{(p,1)}(\Gamma;\Delta) &  =
\mathcal{A}_{(p,1)}(D_{\{1\}} q, D_{\{3\}} p; D_{\{1,2\}} r)
&   \\

  &  =
\langle D_{\{3\}} \rangle \mathcal{A}_{(p,1)} (p;\emptyset)
\lor
 \langle D_{\{2\}} \rangle\mathcal{A}_{(p,1)} ( q; r)
&  \text{from the Line 11} \\

  &  =
\langle D_{\{3\}} \rangle  \bot
\lor
 \langle D_{\{2\}} \rangle (\neg q\lor r)
&  \text{from the Line 11} \\





\end{array}
\]

\end{ex}


\begin{prop}
\label{prop:initial casein main thm in gkn}
Let $\mathbf{L}\in \{ \mathbf{K}_D , \mathbf{KD}_D\}$  Given any multi-sets $\Gamma,\Delta$ of formulas, propositional variables $p$ and $q$, agent symbol $a$  such that $p\neq q$.  We have:
  1, $\mathsf{G}(\mathbf{L})\vdash q\ra \mathcal{A}_{(p,a)}(\Gamma;\Delta,q)$;
2, $\mathsf{G}(\mathbf{L})\vdash \ra \mathcal{A}_{(p,a)}(\Gamma,q;\Delta),q$;
 3, $\mathsf{G}(\mathbf{L})\vdash \ra \mathcal{A}_{(p,a)}(q,\Gamma;\Delta,q)$.
\end{prop}

\begin{thm}
\label{thm:main theorem of gkn}
Let $\mathbf{L}\in \{ \mathbf{K}_D , \mathbf{KD}_D\}$ 
    Let $\Gamma,\Delta$ be finite multi-sets of formulas. For an arbitrary propositional variable $p$ and  an arbitrary agent symbol $a$, there exists a formula $\mathcal{A}_{(p,a)}(\Gamma;\Delta)$ such that:
    \begin{enumerate}[(i)]
        \item  
            $\mathsf{V}( \mathcal{A}_{(p,a)}(\Gamma;\Delta)) \subseteq \mathsf{V}(\Gamma\cup \Delta)\backslash\{p\}    $,  $\mathsf{Agt}( \mathcal{A}_{(p,a)}(\Gamma;\Delta)) \subseteq \mathsf{Agt}(\Gamma\cup \Delta)\backslash\{a\}    $;
        \item $\mathsf{G}(\mathbf{L})\vdash    \Gamma , \mathcal{A}_{(p,a)}(\Gamma; \Delta) \ra \Delta$;
        \item given finite multi-sets $\Pi, \Lambda$ of formulas such that: 
            $p\notin \mathsf{V}(\Pi\cup\Lambda)$, $a\notin \mathsf{Agt}(\Pi\cup\Lambda)$ and $\mathsf{G}(\mathbf{L})\vdash  \Pi,\Gamma \ra \Delta, \Lambda$,
        then 
            $\mathsf{G}(\mathbf{L})\vdash  \Pi \ra  \mathcal{A}_{(p,a)}(\Gamma; \Delta) ,\Lambda$.
    \end{enumerate}
    
\end{thm}

\begin{proof}


The proof of (i) can be obtained by inspecting the table in Definition \ref{dfn:Ap formula in Gkn}.
The proof of  (ii) can be proved by induction on the weight of $\mathcal{A}_{(p,a)}(\Gamma;\Delta)$. We prove $\mathsf{G}(\mathbf{L})\vdash    \Gamma , \mathcal{A}_{(p,a)}(\Gamma; \Delta) \ra \Delta$ for each line of the table in  Definition \ref{dfn:Ap formula in Gkn}.
The cases of lines from 1 to 9 are straightforward. In the line 10 and $10''$, we can have:

\begin{center}
      \begin{scriptsize}
  \AxiomC{$\Gamma'' \ra B $}
\RightLabel{$(D_{K} )$}
\UnaryInfC{$  D_{\{a\}} \Gamma''\ra   D_{\{a\}} B$}
\RightLabel{$(Weakening)^\ast$}
\UnaryInfC{$ \Gamma', D_{\{a\}} \Gamma'',\top\ra  \Delta', D_{\{a\}} B$}
\DisplayProof
  \AxiomC{$\Gamma'' \ra  $}
\RightLabel{$(D_{D} )$}
\UnaryInfC{$  D_{\{a\}} \Gamma''\ra    $}
\RightLabel{$(Weakening)^\ast$}
\UnaryInfC{$ \Gamma', D_{\{a\}} \Gamma'',\top\ra  \Delta $}
\DisplayProof
      \end{scriptsize}
\end{center}

In the case of the line 11,
the idea is to first show the case of each conjunct, then combine them together by $(L\lor)$.
For each $r\in \Phi/\{p\}$, $ \mathsf{G}(\mathbf{L}) \vdash \Phi,\neg r,
      \overrightarrow{D_{G_m}\gamma_m }  \ra \overrightarrow{D_{H_n}\delta_n },\Psi $;
  for each $q\in \Psi/\{p\}$, $ \mathsf{G}(\mathbf{L}) \vdash \Phi,q,
       \overrightarrow{D_{G_m}\gamma_m }  \ra \overrightarrow{D_{H_n}\delta_n },\Psi $;
  for each $\bot\in \Psi$, $ \mathsf{G}(\mathbf{L}) \vdash \Phi,\bot,
       \overrightarrow{D_{G_m}\gamma_m }  \ra \overrightarrow{D_{H_n}\delta_n },\Psi $;
  for each $\bot\in \Phi$, $ \mathsf{G}(\mathbf{L}) \vdash \Phi,\neg \bot,
       \overrightarrow{D_{G_m}\gamma_m }  \ra \overrightarrow{D_{H_n}\delta_n },\Psi $.

\begin{itemize}

    \item for each  $D_{G_j}\gamma_j \in \{\overrightarrow{D_{G_m}\gamma_m }\}
    /\{\overrightarrow{D_{G_m}\gamma_m }\}^{\natural_{\ni a}}   $,

  \begin{scriptsize}
    \begin{center}
        \AxiomC{I.H. }
       \noLine
        \UnaryInfC{$\{ \overrightarrow{D_{G_m\gamma_m}} \}^{\flat_{\subseteq G_j}}
,\mathcal{A}_{(p,a)} (   \{ \overrightarrow{D_{G_m\gamma_m}} \}^{\flat_{\subseteq G_j}} ;\emptyset)  \ra $}
    \RightLabel{$(R\neg)$}
        \UnaryInfC{$\{ \overrightarrow{D_{G_m\gamma_m}} \}^{\flat_{\subseteq G_j}} \ra
\neg \mathcal{A}_{(p,a)} (   \{ \overrightarrow{D_{G_m\gamma_m}} \}^{\flat_{\subseteq G_j}};\emptyset)   $}
    \RightLabel{$(D_{K})$}
        \UnaryInfC{$ \Phi, \overrightarrow{D_{G_m\gamma_m}}   \ra
 D_{G_j} \neg \mathcal{A}_{(p,a)} (   \{ \overrightarrow{D_{G_m\gamma_m}} \}^{\flat_{\subseteq G_j}};\emptyset) , \overrightarrow{D_{H_n}\delta_n },\Psi   $}
    \RightLabel{$(L\neg)$}
        \UnaryInfC{$ \Phi,  \langle D_{G_j} \rangle  \mathcal{A}_{(p,a)} (   \{ \overrightarrow{D_{G_m\gamma_m}} \}^{\flat_{\subseteq G_j}};\emptyset) , \overrightarrow{D_{G_m\gamma_m}} ,  \ra
  \overrightarrow{D_{H_n}\delta_n },\Psi   $}
        \DisplayProof
    \end{center}
\end{scriptsize}


\item  
for each $D_{H_i}\delta_i \in \{\overrightarrow{D_{H_n}\delta_n } \}
  / \{\overrightarrow{D_{H_n}\delta_n } \}^{\natural_{\ni a}}$, by induction hypothesis, we obtain:
 

    \begin{center}
        \AxiomC{I.H. }
       \noLine
        \UnaryInfC{$  \{\overrightarrow{D_{G_m}\gamma_m }  \}^{\flat_{\subseteq{H_i}}}
,\mathcal{A}_{(p,a)} (    \{\overrightarrow{D_{G_m}\gamma_m }  \}^{\flat_{\subseteq{H_i}}} ; \delta_i ) \ra \delta_i  $}
    \RightLabel{$(D_{K})$}
        \UnaryInfC{$ \Phi,  \overrightarrow{D_{G_m}\gamma_m }   ,
        D_{H_i}\mathcal{A}_{(p,a)} (    \{\overrightarrow{D_{G_m}\gamma_m }  \}^{\flat_{\subseteq{H_i}}} ; \delta_i ) \ra  
        \overrightarrow{D_{H_n}\delta_n },\Psi   $}
 \DisplayProof
    \end{center}


\item 
for each $ D_{H_i}\delta_i \in  \{\overrightarrow{D_{H_n}\delta_n } \}^{\natural_{\ni a}}$ and $H_i \ne \{a\}$.
From induction hypothesis,     

\begin{center}
    \AxiomC{I.H. }
    \noLine
    \UnaryInfC{$  \{\overrightarrow{D_{G_m}\gamma_m }  \}^{\flat_{\subseteq{H_i}}}   ,  \mathcal{A}_{(p,a) }(  \{\overrightarrow{D_{G_m}\gamma_m }  \}^{\flat_{\subseteq{H_i}}}; \delta_i  )  \ra \delta_i $}
    \RightLabel{$(D_K)$}
    \UnaryInfC{$\Phi, \overrightarrow{D_{G_m}\gamma_m }    ,  D_{H_i/\{a\}}\mathcal{A}_{(p,a) }(  \{\overrightarrow{D_{G_m}\gamma_m }  \}^{\flat_{\subseteq{H_i}}}; 
    \delta_i  )  \ra   \overrightarrow{D_{H_n}\delta_n },\Psi $ } 
    \DisplayProof
\end{center}

\end{itemize}

Finally,  $(L\lor)$ is applied finitely many times, we obtain the following:

   \begin{tiny}
\[
 \mathsf{G}(\mathbf{L}) \vdash
\bigvee_{r\in\Phi/\{p\}}\neg r 
\vee 
\bigvee_{q\in\Psi/\{p\}}q
\vee
\bigvee_{\bot\in\Phi}\neg\bot
\vee
\bigvee_{\bot\in\Psi}\bot
\vee
\bigvee_{D_{G_j}\gamma_j \in \{\overrightarrow{D_{G_m}\gamma_m }\} 
   / \{\overrightarrow{D_{G_m}\gamma_m }\}^{\natural_{\ni a}}  }
\langle D_{G_j} \rangle
\mathcal{A}_{(p,a)} (    \{ \overrightarrow{D_{G_m\gamma_m}} \}^{\flat_{\subseteq G_j}} ;\emptyset)
\]
\[
  \vee \bigvee_{D_{H_i}\delta_i \in \{\overrightarrow{D_{H_n}\delta_n } \}
  / \{\overrightarrow{D_{H_n}\delta_n } \}^{\natural_{\ni a}} }
  D_{H_i} \mathcal{A}_{(p,a)} (   \{\overrightarrow{D_{G_m}\gamma_m }  \}^{\flat_{\subseteq{H_i}}}; \delta_i   )
\]
\[
  \vee
     \bigvee_{D_{H_i}\delta_i \in  \{\overrightarrow{D_{H_n}\delta_n } \}^{\natural_{\ni a}}
     ~\text{with}~H_i \ne \{a\}}
  D_{H_i/\{a\}} \mathcal{A}_{(p,a)} (   \{\overrightarrow{D_{G_m}\gamma_m }  \}^{\flat_{\subseteq{H_i}}}; \delta_i   ),
  \Phi, \overrightarrow{D_{G_m}\gamma_m }  \ra \overrightarrow{D_{H_n}\delta_n },\Psi
\] 
   \end{tiny}


In (iii), we consider the last rules applied in the derivation of  $  \Pi,\Gamma \ra \Delta, \Lambda$.
  When it is an initial sequent, we need the Proposition \ref{prop:initial casein main thm in gkn}.
When the last rule is among propositional logical rules, the proof is straightforward from the height-preserving invertibility in the Proposition \ref{prop:invertibility of all logical rules in Gkn}.

When the last rule is $(D_{K})$, we assume that the derivation is ended in the form of 
$\Pi, \Gamma \ra \Delta, \Lambda$ and 
 $p\notin \mathsf{V}(\Pi\cup\Lambda)$, $a\notin \mathsf{Agt}(\Pi\cup\Lambda)$.
The arguments are divided into the following cases:

\begin{enumerate}
    \item The right principal formula $D_H \alpha$  is in the multiset $\Lambda$.

 \begin{center}

\AxiomC{$\overrightarrow{\pi_m},\overrightarrow{\gamma_n} \ra \alpha$}
\RightLabel{$(D_{K}) $}
\UnaryInfC{$\Pi', \overrightarrow{D_{P_m}\pi_m},\Phi,\overrightarrow{D_{G'_j} \beta_j}, \overrightarrow{D_{G_n}\gamma_n} \ra
D_H \alpha, \Lambda', \overrightarrow{D_{H'_k} \delta_k}, \Psi $}
\DisplayProof

\end{center}
with $ \bigcup_{\overrightarrow{P_m}}\cup \bigcup_{\overrightarrow{G_n}} \subseteq H
$, 
where,
  $\Pi', \overrightarrow{D_{P_m}\pi_m}$ is $\Pi$, and $\Pi'$ contains propositional variables, $\bot$ and
        outmost-boxed formulas  not $H$-sub;
  $\Phi,\overrightarrow{D_{G'_j} \beta_j}, \overrightarrow{D_{G_n}\gamma_n}$ is $\Gamma$, where
 any group of $\overrightarrow{G'_j}$ is  not the subset of $H$,
      $\Phi$ contains only propositional variables or $\bot$;
  $D_H \alpha, \Lambda'$ is $\Lambda$, and $\Lambda'$ contains propositional variables, $\bot$ and outmost-boxed formula; 
  $\overrightarrow{D_{H'_k} \delta_k}, \Psi $ is $\Delta$, and $\Psi$ contains only propositional variables, $\bot$.
    In this case, $a\notin H$ since $ a\notin \mathsf{Agt}(\Lambda)$.

        \begin{enumerate}
        \item Some formulas in $\Gamma$ are principal, that is $\overrightarrow{\gamma_n}$ is not empty.
        In this case, we obtain that  $p$ and $a$  
        are no in formulas $\overrightarrow{\pi_m},\alpha$ from assumption. After applying induction 
        hypothesis, we obtain


        \begin{center}
        \AxiomC{I.H}
        \noLine
        \UnaryInfC{$\overrightarrow{\pi_m} \ra 
         \mathcal{A}_{(p,a) }(\overrightarrow{\gamma_n;}    \emptyset ),\alpha$}
            \RightLabel{$(L\neg) $}
           \UnaryInfC{$\overrightarrow{\pi_m} ,
        \neg \mathcal{A}_{(p,a) }(\overrightarrow{\gamma_n;}    \emptyset )\ra \alpha$}
                 \RightLabel{$(D_{K}) $}
           \UnaryInfC{$\Pi',\overrightarrow{D_{P_m}\pi_m} ,
        D_{G_i} \neg \mathcal{A}_{(p,a) }(\overrightarrow{\gamma_n;}    \emptyset )\ra D_H \alpha,\Lambda'$}
                 \RightLabel{$(R\neg) $}
        \UnaryInfC{$  \Pi',\overrightarrow{D_{P_m}\pi_m} \ra 
       \langle D_{G_i}\rangle 
    \mathcal{A}_{(p,a) }(\overrightarrow{\gamma_n};
    \emptyset ),
    D_H \alpha, \Lambda'$}

        \DisplayProof
        \end{center}
where $i$ is among $n$ and $a\notin D_{G_i}$, since $a\notin H$ and $D_{G_i}\subseteq H$.
Especially, $\overrightarrow{\gamma_n}$ is equal to $\{\overrightarrow{D_{G'_j} \beta_j},
 \overrightarrow{D_{G_n}\gamma_n} \}^{\flat_{ \subseteq H}}$.
 Then, we have   $\Pi',\overrightarrow{D_{P_m}\pi_m} \ra 
       \langle D_{G_i}\rangle 
    \mathcal{A}_{(p,a) }(\{\overrightarrow{D_{G'_j} \beta_j},
 \overrightarrow{D_{G_n}\gamma_n} \}^{\flat_{ \subseteq H}};
    \emptyset ),
    D_H \alpha, \Lambda'$. After applying $(R\lor)$ and weakening rules for many times, we obtain $\Pi',\overrightarrow{D_{P_m}\pi_m} \ra 
         \mathcal{A}_{(p,a) } ( \Phi,\overrightarrow{D_{G'_j} \beta_j}, \overrightarrow{D_{G_n}\gamma_n} ;
    \overrightarrow{D_{H'_k} \delta_k}, \Psi ),
    D_H \alpha, \Lambda'$.


         \item All formulas in $\Gamma$ are not principal.  That is  $\overrightarrow{\gamma_n}$ is  empty and  principal formulas are all in $\Pi$ and $\Lambda$.
         Since $\Pi\ra \Lambda$ is derivable from assumption, we can derive the desired result by applying a right weakening rule.
    \end{enumerate}

        \item The right principal formula $D_H \alpha$  is in the multiset $\Delta$, and $a\notin H$.

       \begin{center}
    
\AxiomC{$\overrightarrow{\pi_m},\overrightarrow{\gamma_n} \ra \alpha$}
\RightLabel{$(D_{K}) $}
\UnaryInfC{$\Pi', \overrightarrow{D_{P_m}\pi_m},\Phi,\overrightarrow{D_{G'_j} \beta_j}, \overrightarrow{D_{G_n}\gamma_n} \ra D_H \alpha,\overrightarrow{D_{H'_k} \delta_k}, \Psi, \Lambda $}
\DisplayProof

\end{center}
with $ \bigcup_{\overrightarrow{P_m}}\cup \bigcup_{\overrightarrow{G_n}} \subseteq H
$, and any $G'_j \not\subseteq H$, where
 $\Pi', \overrightarrow{D_{P_m}\pi_m}$ is $\Pi$, and $\Pi'$ contains propositional variables, $\bot$ and
        outmost-boxed formulas  not $H$-sub.  For any group of $\overrightarrow{P_m}$, it is a subset of $H/\{a\} $ because $a$ cannot occur in $\Pi$;
  $\Phi,\overrightarrow{D_{G'_j} \beta_j}, \overrightarrow{D_{G_n}\gamma_n}$ is $\Gamma$, where
 any group of $\overrightarrow{G'_j}$ is  not the subset of $H$,
      $\Phi$ contains only propositional variables or $\bot$;
   $D_H \alpha, \overrightarrow{D_{H'_k} \delta_k}, \Psi $ is $\Delta$, and $\Psi$ contains only propositional variables, $\bot$.
From assumption, $p$ and $a$ are not in $\overrightarrow{\pi_m}$. Induction hypothesis gives:

\begin{center}
\AxiomC{I.H.}
\noLine
\RightLabel{}
    \UnaryInfC{$\overrightarrow{\pi_m} \ra  \mathcal{A}_{(p,a) }(\overrightarrow{\gamma_n};\alpha)$ }
\RightLabel{$(D_{K}) $}
\UnaryInfC{$\Pi', \overrightarrow{D_{P_m}\pi_m}  \ra D_{H} 
\mathcal{A}_{(p,a) }(\overrightarrow{\gamma_n};\alpha) ,
  \Lambda $}
\DisplayProof
\end{center}

Then, $\overrightarrow{\gamma_n}$ is equal to $\{\overrightarrow{D_{G'_j} \beta_j},
 \overrightarrow{D_{G_n}\gamma_n} \}^{\flat_{ \subseteq H}}$. We have that 
 $\Pi', \overrightarrow{D_{P_m}\pi_m}  \ra D_{H} \mathcal{A}_{(p,a) }(\{\overrightarrow{D_{G'_j} \beta_j},
 \overrightarrow{D_{G_n}\gamma_n} \}^{\flat_{ \subseteq H}};\alpha), \Lambda $. Then, we apply $(R\lor)$ rules many times.


        \item The right principal formula $D_H \alpha$ is in the multiset $\Delta$, and $a\in H$.

\begin{center}
    
\AxiomC{$\overrightarrow{\pi_m},\overrightarrow{\gamma_n} \ra \alpha$}
\RightLabel{$(D_{K}) $}
\UnaryInfC{$\Pi', \overrightarrow{D_{P_m}\pi_m},\Phi,\overrightarrow{D_{G'_j} \beta_j}, \overrightarrow{D_{G_n}\gamma_n} \ra D_H \alpha,\overrightarrow{D_{H'_k} \delta_k}, \Psi, \Lambda $}
\DisplayProof

\end{center}


with $ \bigcup_{\overrightarrow{P_m}}\cup \bigcup_{\overrightarrow{G_n}} \subseteq H
$, 
where $\Pi', \overrightarrow{D_{P_m}\pi_m}$ is $\Pi$, and $\Pi'$ contains propositional variables, $\bot$ and
        outmost-boxed formulas  not $H$-sub.  For any group of $\overrightarrow{P_m}$, it is a subset of $H/\{a\} $ because $a$ cannot occur in $\Pi$;
      $\Phi,\overrightarrow{D_{G'_j} \beta_j}, \overrightarrow{D_{G_n}\gamma_n}$ is $\Gamma$, where
 any group of $\overrightarrow{G'_j}$ is  not the subset of $H$,
      $\Phi$ contains only propositional variables or $\bot$;
  $D_H \alpha, \overrightarrow{D_{H'_k} \delta_k}, \Psi $ is $\Delta$, and $\Psi$ contains only propositional variables, $\bot$.
Depending on whether $H$ is equal to $\{a\}$, we distinguish the following subcases.

\begin{enumerate}
    \item  When $H\ne \{a\}$,   $p$ and $a$  
        does not occur in formulas $\overrightarrow{\pi_m}$ from assumption. Then, according to induction hypothesis, 

\begin{center}
\AxiomC{I.H.}
\noLine
\RightLabel{}
    \UnaryInfC{$\overrightarrow{\pi_m} \ra  \mathcal{A}_{(p,a) }(\overrightarrow{\gamma_n};\alpha$) }
\RightLabel{$(D_{K}) $}
\UnaryInfC{$\Pi', \overrightarrow{D_{P_m}\pi_m}  \ra D_{H/\{a\}} 
\mathcal{A}_{(p,a) }(\overrightarrow{\gamma_n};\alpha) ,
  \Lambda $}
\DisplayProof
\end{center}

where $ \bigcup_{\overrightarrow{P_m}} \subseteq H/\{a\}$, since $a\notin \mathsf{Agt}(\Pi)$ as we have mentioned above. Also, 
$\overrightarrow{\gamma_n}$ is equal to $\{\overrightarrow{D_{G'_j} \beta_j},
 \overrightarrow{D_{G_n}\gamma_n} \}^{\flat_{ \subseteq H}}$.
 Then, we have   $\Pi',\overrightarrow{D_{P_m}\pi_m} \ra 
     D_{H/\{a\}}
    \mathcal{A}_{(p,a) }(\{\overrightarrow{D_{G'_j} \beta_j},
 \overrightarrow{D_{G_n}\gamma_n} \}^{\flat_{ \subseteq H}};
    \alpha ), \Lambda$. Next, we apply $(R\lor)$ many times to obtain the desired result.

\item
When $H= \{a\}$, $\overrightarrow{\pi}$ must be empty. The derivation is ended as: 

 \begin{center}
    
\AxiomC{$ \overrightarrow{\gamma_n} \ra \alpha$}
\RightLabel{$(D_{K}) $}
\UnaryInfC{$\Pi , \Phi,\overrightarrow{D_{G_j} \beta_j}, \overrightarrow{D_{\{a\}}\gamma_n} 
\ra D_{\{a\}} \alpha,\overrightarrow{D_{H'_k} \delta_k}, \Psi,  \Lambda $}
\DisplayProof

\end{center}
Since $ \overrightarrow{\gamma_n} \ra \alpha$ is derivable, $\mathcal{A}_{(p,a) }(\Phi,\overrightarrow{D_{G_j} \beta_j}, \overrightarrow{D_{\{a\}}\gamma_n} 
; D_{\{a\}} \alpha,\overrightarrow{D_{H'_k} \delta_k}, \Psi)$  is $\top$
from  the definition \ref{dfn:Ap formula in Gkn} (the line 10).
  It is obvious that $\mathsf{G}(\mathbf{L}) \vdash \Pi\ra \top, \Lambda.$
\end{enumerate}


       \end{enumerate}

The last part to examine the conditions of (iii) is 
 the case  of $(D_D)$ of $\gkdd$. When the derivation of $\Pi,\Gamma\ra \Delta,\Lambda$ is ended with:

\begin{center}
    
\AxiomC{$\overrightarrow{\pi_m},\overrightarrow{\gamma_n} \ra $}
\RightLabel{$(D_{D}) $}
\UnaryInfC{$\Pi', D_{\{b\}}\overrightarrow{\pi_m},\Phi,\overrightarrow{D_{G'_j} \beta_j}, D_{\{b\}}\overrightarrow{\gamma_n} \ra 
\Delta, \Lambda$
}
\DisplayProof
\end{center}
where  $\Pi', D_{\{b\}}\overrightarrow{\pi_m}$ is $\Pi$, and $\Pi'$ contains propositional variables, $\bot$ and
        outmost-boxed formulas  not $\{b\}$-sub;   
  $\Phi,\overrightarrow{D_{G'_j} \beta_j}, D_{\{b\}}\overrightarrow{\gamma_n}$ is $\Gamma$, where
 any group of $\overrightarrow{G'_j}$ is  not the subset of $\{b\}$;
      $\Phi$ contains only propositional variables or $\bot$.
Depending on whether $b$ is equal to $a$, we consider the following two cases:
\begin{enumerate}
    \item When $b\ne a$, we distinguish two subcases:
    \begin{enumerate}
    
        \item If $\overrightarrow{\gamma_n}$ is not empty,        
     $p$ and $a$ do not occur in $\overrightarrow{\pi_m}$ according to the assumption.
    Then, by induction hypothesis, we obtain:

    \begin{center}
    
\AxiomC{ $\overrightarrow{\pi_m} \ra 
         \mathcal{A}_{(p,a) }(\overrightarrow{\gamma_n;}    \emptyset ) $}
            \RightLabel{$(L\neg) $}
           \UnaryInfC{$\overrightarrow{\pi_m} ,
        \neg \mathcal{A}_{(p,a) }(\overrightarrow{\gamma_n;}    \emptyset )\ra $}
                 \RightLabel{$(D_{K}) $}
           \UnaryInfC{$\Pi',D_{\{b\}}\overrightarrow{\pi_m} ,
        D_{\{b\}} \neg \mathcal{A}_{(p,a) }(\overrightarrow{\gamma_n;}    \emptyset )\ra \Lambda$}
                 \RightLabel{$(R\neg) $}
\RightLabel{$(R\neg) $}
\UnaryInfC{$\Pi',D_{\{b\}}\overrightarrow{\pi_m} 
       \ra  \langle D_{\{b\}}\rangle  \mathcal{A}_{(p,a) }(\overrightarrow{\gamma_n;}    \emptyset ) ,\Lambda$}
 
\DisplayProof
\end{center}

    Then, $\overrightarrow{\gamma_n}$ is equal to $ \{\overrightarrow{D_{G'_j} \beta_j},
 \ \overrightarrow{D_{\{b\}}\gamma_n}\}^{\flat_{ \subseteq \{b\}}}$.
That is, $\Pi', D_{\{b\}}\overrightarrow{\pi_m}, \ra \langle D_{\{b\}}\rangle  \mathcal{A}_{(p,a) }(\{\overrightarrow{D_{G'_j} \beta_j},
 \ \overrightarrow{D_{\{b\}}\gamma_n}\}^{\flat_{ \subseteq \{b\}}} ;    \emptyset ),
  \Lambda$ is derivable in $\gkdd$.
  At last, we apply $(R\lor)$ many times, we obtain that 
  $\Pi', D_{\{b\}}\overrightarrow{\pi_m}, \ra   \mathcal{A}_{(p,a) }( \Phi,\overrightarrow{D_{G'_j} \beta_j}, D_{\{b\}}\overrightarrow{\gamma_n}; \Delta ),  \Lambda$
    
 \item If $\overrightarrow{\gamma_n}$ is empty, then we obtain:

 \begin{center}
    
\AxiomC{$\overrightarrow{\pi_m}  \ra $}
\RightLabel{$(D_{D}) $}
\UnaryInfC{$\Pi', D_{\{b\}}\overrightarrow{\pi_m}  \ra 
 \Lambda$
}
\DisplayProof
\end{center}
Then, we apply $(RW)$ to obtain the desired result.

     \end{enumerate}

    \item When $b=a$, $\overrightarrow{\pi_m} $ is empty since $a$ does not occur in $\Pi$.
    It implies that $\overrightarrow{\gamma_n} \ra $ is derivable. According to  the definition \ref{dfn:Ap formula in Gkn} (the line $10'$ ), 
    $\mathcal{A}_{(p,a) }( \Phi,\overrightarrow{D_{G'_j} \beta_j}, D_{\{b\}}\overrightarrow{\gamma_n}; \Delta )$ is $\top$.
It is evident that     $\gkdd\vdash \Pi\ra \top,\Lambda$.\qedhere
\end{enumerate}
\end{proof}

\section{UIP in Logic $\mathbf{KT}_D$}
\label{sec:T}

\subsection{Proof-theoretic properties of sequent calculus}

The possible loop in derivation of $\gktd$ will bring difficulties in defining $\mathcal{A}$-formula.
In the following, a sequent calculus with
built-in loop-check mechanism will be presented.
This sequent is an expansion of single modal calculus from   \cite{bilkova2007uniform}, which is 
 inspired by the work from
\cite{Heuerding1996}.
A {\bf T}-sequent, denoted by $\Sigma|\Gamma \ra \Delta$ is obtained from adding a finite multiset $\Sigma$  into a sequent $\Gamma\ra \Delta$, where $\Sigma$ containing only outmost-boxed formulas.

\begin{table}[htb]
\caption{System $\gktdplus$.}
\label{table:gktnplus}
\begin{footnotesize}
\hrule
\begin{tabular}{ll}

\multicolumn{2}{c}{Sequent Calculus $\gktdplus$: } \\


{ \bf Initial Sequents}   & $\Sigma|\Gamma ,p \ra p, \Delta$\hspace{15pt} $\Sigma|\bot, \Gamma\ra \Delta$  \\ 
 \\







{\bf Logical Rules} & 

\AxiomC{$\Sigma|\Gamma \ra \Delta ,\alpha_1$}
\AxiomC{$\Sigma|\Gamma \ra \Delta ,\alpha_2$}
\RightLabel{\scriptsize $(R\wedge )$}
\BinaryInfC{$\Sigma|\Gamma \ra  \Delta ,\alpha_1\wedge \alpha_2$}
\DisplayProof

\AxiomC{$\Sigma|\alpha_1, \alpha_2,\Gamma \ra \Delta$}
\RightLabel{\scriptsize $(L\wedge )$}
\UnaryInfC{$\Sigma|\alpha_1\wedge \alpha_2,\Gamma \ra \Delta$}
\DisplayProof
\\

\,  &

\AxiomC{$\Sigma|\Gamma \ra \Delta, \alpha_1, \alpha_2$}
\RightLabel{\scriptsize $(R\lor )$}
\UnaryInfC{$\Sigma|\Gamma \ra \Delta,\alpha_1\lor \alpha_2$}
\DisplayProof

\AxiomC{$\Sigma|\alpha_1,\Gamma \ra \Delta$}
\AxiomC{$\Sigma|\alpha_2,\Gamma \ra \Delta$}
\RightLabel{\scriptsize $(L\lor )$}
\BinaryInfC{$\Sigma|\alpha_1\lor \alpha_2, \Gamma \ra \Delta$}
\DisplayProof

\\

\,  &

\AxiomC{$\Sigma|\alpha_1,\Gamma \ra \Delta, \alpha_2$}
\RightLabel{\scriptsize $(R\rightarrow )$}
\UnaryInfC{$\Sigma|\Gamma \ra \Delta,\alpha_1\rightarrow \alpha_2$}
\DisplayProof

\AxiomC{$\Sigma|\Gamma \ra \Delta, \alpha_1$}
\AxiomC{$\Sigma|\alpha_2,\Gamma \ra \Delta$}
\RightLabel{\scriptsize $(L\rightarrow )$}
\BinaryInfC{$\Sigma|\alpha_1\rightarrow \alpha_2, \Gamma \ra \Delta$}
\DisplayProof

\\

\,  &
\AxiomC{$\Sigma|\alpha, \Gamma \ra  \Delta$}
\RightLabel{\scriptsize $(R\neg )$}
\UnaryInfC{$\Sigma|\Gamma \ra \Delta,\neg \alpha$}
\DisplayProof

\AxiomC{$\Sigma|\Gamma \ra \Delta, \alpha$}
\RightLabel{\scriptsize $(L\neg )$}
\UnaryInfC{$\Sigma|\neg \alpha, \Gamma \ra  \Delta$}
\DisplayProof
\\




   


{\bf Modal Rule}   &

\AxiomC{$\emptyset|\alpha_1,\ldots,\alpha_n \ra \beta  $ }
\RightLabel{{\scriptsize $(D_{K}^+)$}$\dagger$  for any $i$ $(0\leq i \leq n)$, $G_i\subseteq G$  }
\UnaryInfC{$\Sigma,   D_{G_1}\alpha_1,\ldots,D_{G_n}\alpha_n |\Pi  \ra D_G \beta, \Omega$ }
\DisplayProof
\\
 \\

\, & 
\AxiomC{$ D_G \alpha, \Sigma|  \Gamma, \alpha \ra \Delta $}
\RightLabel{\scriptsize $(D_{T}^+)$}
\UnaryInfC{$\Sigma |\Gamma, D_G \alpha    \ra \Delta$}
\DisplayProof \\



\multicolumn{2}{l}{ {\footnotesize $\dagger$: $\Sigma$ contains only outmost-boxed formulas which are not $G$-sub formulas, $\Pi$ contains only propositional variables, $\bot$,  $\Omega$ contains only  }
   } \\
   
\multicolumn{2}{l}{ {\footnotesize  propositional variables, $\bot$ or any outmost-boxed formulas. }
 }\\
\hline

\end{tabular}
\end{footnotesize}
\end{table}






    

    In  $\gktdplus$, we say that the   multisets  in $\Sigma,\Gamma,\Delta, \Omega$ and $\Pi$ are 
     {\it context} in all rules. A formula (or multiset) is called {\it principal} in a rule if it is not in the context.
The definition of derivation, (height-preserving) admissibility and (height-preserving) invertibility are defined similarly as $\gktd$.

\begin{defn}
    Let $\alpha$ be a formula,  $\mathsf{b} (\alpha)$ be the number of outmost boxed subformulas in  $A$. Given a {\it set} $\Gamma$, $\mathsf{b} (\Gamma)$ denotes the sum of all $\mathsf{b} (\alpha)$ for  $\alpha\in \Gamma$. 
Given multi-sets $\Gamma,\Delta,\Gamma^\prime,\Delta^\prime$ of formulas:
    $\langle\mathsf{b}(\Gamma),\mathsf{wt}(\Delta)\rangle <
\langle\mathsf{b}(\Gamma^\prime),\mathsf{wt}(\Delta^\prime)\rangle$
denotes a lexicographical order on a pair of natural number. 
We define a well-ordered relation of {\bf T}-sequent:
$(\Sigma|\Gamma\ra \Delta)\prec (\Sigma^\prime|\Gamma^\prime\ra \Delta^\prime) $ if and only if 
$\langle \mathsf{b}( \Sigma,\Gamma, \Delta),\mathsf{wt}(\Gamma,\Delta)\rangle    < 
\langle \mathsf{b}( \Sigma^\prime,\Gamma^\prime, \Delta^\prime),\mathsf{wt}(\Gamma^\prime,\Delta^\prime)\rangle$.     
\end{defn}

\begin{prop}
\label{prop:gktplus terminate}
A backward proof-searching in $\gktdplus$ always terminates.
\end{prop}
\begin{proof}
   Consider a rule $(\circ)$ in $\gktdplus$ as

\begin{center}
    \AxiomC{$\Sigma|\Gamma\ra \Delta$}
        \AxiomC{$\dots$}
    \RightLabel{\scriptsize $(\circ)$}
    \BinaryInfC{$\Sigma^\prime|\Gamma^\prime\ra \Delta^\prime $}
    \DisplayProof
\end{center}
We show that   for any premise $\Sigma|\Gamma\ra \Delta$ and the conclusion $\Sigma^\prime|\Gamma^\prime\ra \Delta^\prime $, $(\Sigma|\Gamma\ra \Delta)\prec (\Sigma^\prime|\Gamma^\prime\ra \Delta^\prime) $. 
    Let  $(\circ)$ be an arbitrary logical rule.  We have $\mathsf{b}( \Sigma ,\Gamma , \Delta )=\mathsf{b}( \Sigma^\prime,\Gamma^\prime, \Delta^\prime)$, however $\mathsf{wt}(\Gamma,\Delta)<\mathsf{wt}(\Gamma^\prime,\Delta^\prime)$.
   Let  $(\circ)$ be   $(D_{T}^+)$. We still have   $\mathsf{b}( \Sigma ,\Gamma , \Delta )=\mathsf{b}( \Sigma^\prime,\Gamma^\prime, \Delta^\prime)$, and $\mathsf{wt}(\Gamma,\Delta)<\mathsf{wt}(\Gamma^\prime,\Delta^\prime)$. 
 Let  $(\circ)$ be $(D_{K}^+)$. We have $\mathsf{b}( \Sigma ,\Gamma , \Delta )<\mathsf{b}( \Sigma^\prime,\Gamma^\prime, \Delta^\prime)$.
\end{proof}

\begin{prop} The following weakening rules are  admissible in $\gktdplus$.

    \begin{center}
\begin{tabular}{lll} 

\AxiomC{$\Sigma|\Gamma \ra \Delta$}
\RightLabel{\scriptsize $(RW)$}
\UnaryInfC{$\Sigma|\Gamma \ra \Delta , \alpha$}
\DisplayProof

   & 

\AxiomC{$\Sigma|\Gamma \ra \Delta$}
\RightLabel{\scriptsize $(LW)$}
\UnaryInfC{$\Sigma|\alpha,\Gamma \ra \Delta$}
\DisplayProof
&
\AxiomC{$\Sigma|\Gamma \ra \Delta$}
\RightLabel{\scriptsize $(LW^+)$}
\UnaryInfC{$\Sigma,D_G  \alpha|\Gamma \ra \Delta$}
\DisplayProof

\end{tabular}
\end{center}
\end{prop}



\begin{prop} 
\label{prop:inverti of t plus rule}
All rules in $\gktdplus$ except $(D^+_{K})$ are height-preserving invertible.
    
\end{prop}

\begin{prop} 
    \label{prop:hp contraction of t plus}
    The following contraction rules are  height-preserving admissible in $\gktdplus$.

  \begin{center}

    \begin{scriptsize}
\begin{tabular}{lll}

            \AxiomC{$\Sigma|\Gamma \ra \Delta , A,A$}
\RightLabel{\scriptsize $(RC)$}
\UnaryInfC{$\Sigma|\Gamma \ra \Delta ,A $}
\DisplayProof

&

\AxiomC{$\Sigma|A,A,\Gamma \ra \Delta$}
\RightLabel{\scriptsize $(LC)$}
\UnaryInfC{$\Sigma|A,\Gamma \ra \Delta$}
\DisplayProof

&
\AxiomC{$\Sigma,D_G A,D_G A|\Gamma \ra \Delta$}
\RightLabel{\scriptsize $(LC^+)$}
\UnaryInfC{$\Sigma,D_G A|\Gamma \ra \Delta$}
\DisplayProof

\end{tabular}

\end{scriptsize}
    \end{center}
    
\end{prop}


 An expected form of cut in the form of:

\begin{center}
 \AxiomC{$\Sigma|\Gamma \ra  \Delta, \lambda $}
\AxiomC{$\Sigma^\prime |\lambda,\Gamma ^{\prime} \ra \Delta ^{\prime} $}
\RightLabel{\scriptsize $(Cut^{\prime})$}
\BinaryInfC{$\Sigma^\prime, \Sigma|\Gamma,\Gamma^{\prime} \ra \Delta, \Delta  ^{\prime} $}
\DisplayProof
\end{center}
is not admissible as pointed out in  \cite[Lemma 4.8]{bilkova2007uniform}. However, since our goal is to show the uniform interpolation, the admissibility of the following forms of cut rules are sufficient.

\begin{prop} The following cut rules are admissible in $\gktdplus$.

\begin{center}
\begin{tabular}{lll}

 \AxiomC{$\emptyset|\Gamma \ra  \Delta, \lambda$}
\AxiomC{$\emptyset|\lambda,\Gamma ^{\prime} \ra \Delta ^{\prime} $}
\RightLabel{\scriptsize $(Cut_1)$}
\BinaryInfC{$\emptyset|\Gamma,\Gamma^{\prime} \ra \Delta, \Delta  ^{\prime} $}
\DisplayProof

&

 \AxiomC{$\Sigma|\Gamma \ra  \Delta, D_G \lambda$}
\AxiomC{$D_G \lambda,\Sigma^\prime |\Gamma ^{\prime} \ra \Delta ^{\prime} $}
\RightLabel{\scriptsize $(Cut_2)$}
\BinaryInfC{$\Sigma^\prime, \Sigma|\Gamma,\Gamma^{\prime} \ra \Delta, \Delta  ^{\prime} $}
\DisplayProof

\end{tabular}
\end{center}
\end{prop}

\begin{proof}
We simultaneously prove these results by following a similar argument in Proposition \ref{prop:cut admissible of k,d,t}. Only the following case will be explained here, other cases are similar.
We consider $\mathtt{rule} (\mathcal{D}_1)$ is $(D_{K}^+)$, 
$\mathtt{rule} (\mathcal{D}_2)$ is also $(D_{K}^+)$ and cut formulas are principal in both rules.

\begin{center}
    \begin{scriptsize}
\noLine
\AxiomC{$\mathcal {D}_1$}
\UnaryInfC{$\emptyset|\overrightarrow{\gamma_m}\ra \lambda$}
\RightLabel{\scriptsize $(D_{K}^+)$}
\UnaryInfC{$\Sigma,\overrightarrow{D_{G_m}\gamma_m}| \Pi \ra D_{G}\lambda,\Omega$}
\noLine
\AxiomC{$\mathcal {D}_2$}
\UnaryInfC{$\emptyset|\lambda ,\overrightarrow{\gamma'_n} \ra  \beta $}
\RightLabel{\scriptsize $(D_{K}^+)$}
\UnaryInfC{$D_{G}\lambda,\Sigma',\overrightarrow{D_{H_n}\gamma'_n}| \Pi' \ra D_{H}\beta,\Omega'$}
\RightLabel{\scriptsize $(cut_2)$}
\BinaryInfC{$\Sigma, \Sigma^{\prime},  \overrightarrow{D_{G_m}\gamma_m} , \overrightarrow{D_{H_n}\gamma'_n} | \Pi,\Pi' \ra  D_{H}\beta,\Omega,\Omega^{\prime}$}
\DisplayProof
    \end{scriptsize}
\end{center}

Then we can transform the derivation into the following:

\begin{center}
    \begin{scriptsize}
\noLine
\AxiomC{$\mathcal {D}_1$}
\UnaryInfC{$\emptyset|\overrightarrow{\gamma_m}\ra \lambda$} 
\noLine
\AxiomC{$\mathcal {D}_2$}
\UnaryInfC{$\emptyset|\lambda ,\overrightarrow{\gamma'_n} \ra  \beta $} 
\RightLabel{\scriptsize $(cut_1)$}
\BinaryInfC{$\emptyset|\overrightarrow{\gamma_m} ,\overrightarrow{\gamma'_n} \ra  \beta $}
\RightLabel{\scriptsize $(D_{K}^+)$}
\UnaryInfC{$ \Sigma^{\prime},  \overrightarrow{D_{G_m}\gamma_m} , \overrightarrow{D_{H_n}\gamma'_n} | \Pi,\Pi' \ra  D_{H}\beta,\Omega,\Omega^{\prime}$}
\RightLabel{\scriptsize $(L{W}^+)^\ast$}
\UnaryInfC{$\Sigma, \Sigma^{\prime},  \overrightarrow{D_{G_m}\gamma_m} , \overrightarrow{D_{H_n}\gamma'_n} | \Pi,\Pi' \ra  D_{H}\beta,\Omega,\Omega^{\prime}$}
\DisplayProof
    \end{scriptsize}
\end{center}
In the transformed derivation, the application of $(cut_1)$ can be eliminated owing to the lower complexity of the cut formula.
The inclusion conditions can be easily checked.
\end{proof}


By observing all rules of $\gktdplus$ and $\gktd$ we  obtain the following result.
\begin{prop}
    \label{prop:empty sigma}
        Given multi-sets $\Gamma,\Delta$ of formulas in $\mathcal{L}$:
         if $ \gktdplus\vdash \Sigma|\Gamma\ra \Delta$ then $\gktd\vdash \Sigma,\Gamma\ra \Delta$.
\end{prop}

\begin{lem}
\label{lem:equipplloent of gktn and gktnplus}
    Given multi-sets $\Gamma,\Delta$ of formulas in $\mathcal{L}$:
        $\gktd\vdash \Gamma\ra \Delta$ if and only if $\gktdplus\vdash \emptyset|\Gamma\ra \Delta$.
\end{lem}
\begin{proof}
    The right-to-left direction can be obtained from the Proposition \ref{prop:empty sigma}.
    The left-to-right direction can be shown by induction on the derivation. 
\end{proof}

\subsection{Main theorem of $\gktd$}
\label{sec:main thm of T}

Then, similar to Definition \ref{dfn:Ap formula in Gkn}, we can define $\mathcal{A}$-formulas in $\mathbf{T}$-sequent as follows.

\begin{defn}
We say that a $\mathbf{T}$-sequent $\Sigma|\Gamma\ra\Delta$ is {\it critical}, if $\Gamma$ contains only propositional variables and $\bot$, $\Sigma$ contains only outmost boxed formulas and $\Delta$ contains only  propositional variables, $\bot$ and outmost boxed formulas
\end{defn}

\begin{defn}
\label{dfn:ap formula for t}
      Let $\Gamma,\Delta$ be finite multi-sets of formulas,  $\Sigma$ be a finite multi-set of out-most boxed formulas,  $p$ be a propositional variable, $a$ be an agent symbol.
      An $\mathcal{A}$-formula $\mathcal{A}_{(p,a)}(\Sigma|\Gamma;\Delta)$  is defined inductively as follows.

  \begin{scriptsize}
\[
\begin{array} 
{|c|c|c|c|}\hline & \Sigma|\Gamma\ra\Delta ~\text{is not critical and }~ \Sigma|\Gamma;\Delta\text{ matches} & \mathcal{A}_{(p,a)}(\Sigma|\Gamma;\Delta)\text{ equals} \\
\hline






~~1~~ & \Sigma|\Gamma',\alpha_1\wedge \alpha_2;\Delta & \mathcal{A}_{(p,a)}(\Sigma|\Gamma',\alpha_1,\alpha_2;\Delta) \\

2 & \Sigma|\Gamma;\alpha_1\wedge \alpha_2,\Delta' & \mathcal{A}_{(p,a)}(\Sigma|\Gamma;\alpha_1,\Delta')\land \mathcal{A}_{(p,a)}(\Sigma|\Gamma;\alpha_2,\Delta') \\

3 & \Sigma|\Gamma',\alpha_1\vee \alpha_2;\Delta & \mathcal{A}_{(p,a)}(\Sigma|\Gamma',\alpha_1;\Delta)\land \mathcal{A}_{(p,a)}(\Sigma|\Gamma',\alpha_2;\Delta) \\

~~4~~ & \Sigma|\Gamma;\alpha_1\vee \alpha_2,\Delta' & \mathcal{A}_{(p,a)}(\Sigma|\Gamma';\alpha_1,\alpha_2,\Delta') \\

~~5~~ & \Sigma|\Gamma',\neg \alpha;\Delta & \mathcal{A}_{(p,a)}(\Sigma|\Gamma';\alpha,\Delta) \\

6 & \Sigma|\Gamma;\neg \alpha,\Delta' & \mathcal{A}_{(p,a)}(\Sigma|\Gamma,\alpha;\Delta') \\

7 & \Sigma|\Gamma',\alpha_1\rightarrow \alpha_2;\Delta & \mathcal{A}_{(p,a)}(\Sigma|\Gamma';\Delta,\alpha_1)\land \mathcal{A}_p(\Sigma|\Gamma',\alpha_2;\Delta) \\
8 & \Sigma|\Gamma;\alpha_1\rightarrow \alpha_2,\Delta' & \mathcal{A}_{(p,a)}(\Sigma|\Gamma,\alpha_1; \Delta',\alpha_2) \\

9 & \Sigma|\Gamma',D_G \alpha;\Delta & \mathcal{A}_{(p,a)}(\Sigma,D_G \alpha|\Gamma',\alpha;\Delta) \\

\hline
\end{array}
\]

\[
\begin{array}
{|c|c|c|c|}
\hline 
& \Sigma|\Gamma\ra\Delta ~\text{is critical and }~ \Sigma|\Gamma;\Delta\text{ matches} & \mathcal{A}_{(p,a)}(\Sigma|\Gamma;\Delta)\text{ equals} \\
\hline




~~10~~ & \Sigma|\Gamma',p;\Delta',p & \top \\

~~11~~ &  \Sigma,  D_{\{a\}} \Gamma''|  \Gamma'   ; \Delta', D_{\{a\}} \alpha ~\text{with}~   \gktdplus\vdash \emptyset|  \Gamma'' \ra \alpha&  \top \\

12 \ddagger& ~~ \overrightarrow{D_{G_m}\gamma_m }  |\Phi; \overrightarrow{D_{H_n}\delta_n }, \Psi ~~  &  \mathsf{X} \\

\hline
\end{array}
\]
\end{scriptsize}

{\footnotesize 
$\ddagger:$ $\Phi$ and  $\Psi$ are multisets containing  propositional variables or $\bot$.
$\Phi\cup\overrightarrow{D_{G_m}\gamma_m } \cup \overrightarrow{D_{H_n}\delta_n }\cup \Psi$is not empty.}


 
The formula $\mathsf{X}$ is:
  \begin{scriptsize}
\[
\bigvee_{r\in\Phi/\{p\}}\neg r 
\vee 
\bigvee_{q\in\Psi/\{p\}}q
\vee
\bigvee_{\bot\in\Phi}\neg\bot
\vee
\bigvee_{\bot\in\Psi}\bot
\vee
\bigvee_{D_{G_j}\gamma_j \in \{\overrightarrow{D_{G_m}\gamma_m }\} 
   / \{\overrightarrow{D_{G_m}\gamma_m }\}^{\natural_{\ni a}}  }
\langle D_{G_j} \rangle
\mathcal{A}_{(p,a)} ( \emptyset|   \{ \overrightarrow{D_{G_m\gamma_m}} \}^{\flat_{\subseteq G_j}} ;\emptyset) 
\]

\[
  \vee \bigvee_{D_{H_i}\delta_i \in \{\overrightarrow{D_{H_n}\delta_n } \}
  / \{\overrightarrow{D_{H_n}\delta_n } \}^{\natural_{\ni a}} }
  D_{H_i} \mathcal{A}_{(p,a)} (  \emptyset|  \{\overrightarrow{D_{G_m}\gamma_m }  \}^{\flat_{\subseteq{H_i}}}; \delta_i   )
\]
\[
  \vee
     \bigvee_{D_{H_i}\delta_i \in  \{\overrightarrow{D_{H_n}\delta_n } \}^{\natural_{\ni a}}
     ~\text{with}~H_i \ne \{a\}}
  D_{H_i/\{a\}} \mathcal{A}_{(p,a)} (  \emptyset|  \{\overrightarrow{D_{G_m}\gamma_m }  \}^{\flat_{\subseteq{H_i}}}; \delta_i   )
\]
  \end{scriptsize}


The formula $\mathcal{A}_{(p,a)}(\Sigma|\Gamma;\Delta)$ is defined
in the following procedure: at first, the lines $1-9$ are repeatedly applied until $\Sigma|\Gamma\ra \Delta$ reaches a critical sequent. 
Next, we test whether it   matches the line $10$, $11$.
If no one is the case, the line $12$ is applied. We repeat the above procedure until $\Sigma|\Gamma;\Delta$ cannot  match any lines in the table, in this case $\mathcal{A}_{(p,a)}(\Sigma|\Gamma;\Delta)$ is defined as $\bot$.
If $X$ does not match any conditions, $X$ is $\bot$.
 Especially $\mathcal{A}_{(p,a)}(\Sigma|\emptyset;\emptyset)$ is defined as $\bot$.

We can define a well-order relation of $\mathcal{A}$-formulas as follows:
$\mathcal{A}_{(p,a)} (\Sigma|\Gamma; \Delta)\prec \mathcal{A}_{(p,a)} (\Sigma^\prime|\Gamma^\prime; \Delta^\prime) $ if and only if
$(\Sigma|\Gamma\ra \Delta)\prec (\Sigma^\prime|\Gamma^\prime\ra \Delta^\prime) $.
Given the fact that all back proof-search in $\gktdplus$ always terminates (in Proposition \ref{prop:gktplus terminate}), we can see that 
such a formula can always be determined.

\end{defn}

\begin{thm}
\label{thm:main theorem of gktnplus}
   Let $\Gamma,\Delta$ be finite multi-sets of formulas,  $\Sigma$ be a finite multi-set of out-most boxed formulas.
    For every propositional variable $p$ and an arbitrary agent symbol $a$ there exists a formula $\mathcal{A}_{(p,a)}(\Sigma|\Gamma;\Delta)$ such that:
    \begin{enumerate}[(i)]
        \item  
            $\mathsf{V}( \mathcal{A}_{(p,a)}(\Sigma|\Gamma,\Delta)) \subseteq \mathsf{V}(\Sigma\cup\Gamma\cup \Delta)\backslash\{p\}    $,
               $\mathsf{Agt}( \mathcal{A}_{(p,a)}(\Sigma|\Gamma,\Delta)) \subseteq \mathsf{Agt}(\Sigma\cup\Gamma\cup \Delta)\backslash\{a\}    $
        \item $\gktdplus\vdash    \Sigma|\Gamma , \mathcal{A}_{(p,a)}(\Sigma|\Gamma; \Delta) \ra \Delta$
        \item given finite multi-sets $\Pi, \Lambda$ of formulas, $\Theta$  of out-most boxed formulas, such that 
            $p\notin \mathsf{V}(\Pi\cup\Lambda\cup\Theta)$, $a\notin \mathsf{Agt}(\Pi\cup\Lambda\cup\Theta)$ and $\gktdplus \vdash \Theta,\Sigma|  \Pi,\Gamma \ra \Delta, \Lambda$
        then 
            $\gktdplus\vdash  \emptyset|\Theta,\Pi \ra  \mathcal{A}_{(p,a)}(\Sigma|\Gamma; \Delta) ,\Lambda$

    \end{enumerate}
    
\end{thm}

\begin{proof}

We proceed similarly to the proof of Theorem \ref{thm:main theorem of gkn}.
The proof of (i) can be obtained by inspecting the table in Definition \ref{dfn:ap formula for t}.
The proof of  (ii) can be proved by induction on the weight of $\mathcal{A}_{(p,a)}(\Sigma|\Gamma,\Delta)$. We prove $\gktdplus\vdash   \Sigma| \Gamma , \mathcal{A}_{(p,a)}(\Sigma|\Gamma; \Delta) \ra \Delta$ for each line of the table in  Definition \ref{dfn:ap formula for t}.
The cases of lines from 1 to 11 are easy.
We only concentrate on the case of  line 12.

The proof is similar to Theorem \ref{thm:main theorem of gkn}, only significant cases are shown here.


 


\begin{itemize}
    \begin{scriptsize}
    \item for each  $D_{G_j}\gamma_j \in \{\overrightarrow{D_{G_m}\gamma_m }\}
    /\{\overrightarrow{D_{G_m}\gamma_m }\}^{\natural_{\ni a}}   $, we have  $\emptyset|\{ \overrightarrow{D_{G_m\gamma_m}} \}^{\flat_{\subseteq G_j}},\mathcal{A}_{(p,a)} (  \emptyset| \{ \overrightarrow{D_{G_m\gamma_m}} \}^{\flat_{\subseteq G_j}} ;\emptyset)  \ra $ is derivable from induction hypothesis. After applying $(R\neg)$, $(D_{K}^+)$, $(L\neg)$ we obtain that 
    $ \overrightarrow{D_{G_m\gamma_m}}  |\Phi,  \langle D_{G_j} \rangle  \mathcal{A}_{(p,a)} (   \{ \overrightarrow{D_{G_m\gamma_m}} \}^{\flat_{\subseteq G_j}};\emptyset)   \ra
  \overrightarrow{D_{H_n}\delta_n },\Psi   $ is derivable.



\item for each $D_{H_i}\delta_i \in \{\overrightarrow{D_{H_n}\delta_n } \}
  / \{\overrightarrow{D_{H_n}\delta_n } \}^{\natural_{\ni a}}$, by induction hypothesis, we obtain:
 $ \emptyset|\{\overrightarrow{\Box_{g_m}\gamma_m }  \}^{\flat_{di}}, \mathcal{A}_{(p,a)} ( \emptyset|  \{\overrightarrow{\Box_{g_m}\gamma_m }  \}^{\flat_{di}}; \delta_i   )  \ra \delta_i  .$ Then after applying $(D_{K}^+)$, weakening rules and $(D_{T}^+)$, we obtain
$ \overrightarrow{D_{G_m}\gamma_m }  | \Phi  ,
        D_{H_i}\mathcal{A}_{(p,a)} (  \emptyset|  \{\overrightarrow{D_{G_m}\gamma_m }  \}^{\flat_{\subseteq{H_i}}} ; \delta_i ) \ra  
        \overrightarrow{D_{H_n}\delta_n },\Psi   $.



\item for each $ D_{H_i}\delta_i \in  \{\overrightarrow{D_{H_n}\delta_n } \}^{\natural_{\ni a}}$ and $H_i \ne \{a\}$.
From induction hypothesis,  we have $  \emptyset|\{\overrightarrow{D_{G_m}\gamma_m }  \}^{\flat_{\subseteq{H_i}}}   ,  \mathcal{A}_{(p,a) }(\emptyset|  \{\overrightarrow{D_{G_m}\gamma_m }  \}^{\flat_{\subseteq{H_i}}}; \delta_i  )  \ra \delta_i $ is derivable.
Then after applying $(D_{K}^+)$, weakening rules and $(D_{T}^+)$, we obtain
$\overrightarrow{D_{G_m}\gamma_m } |\Phi     ,  D_{H_i/\{a\}}\mathcal{A}_{(p,a) }(  \{\overrightarrow{D_{G_m}\gamma_m }  \}^{\flat_{\subseteq{H_i}}}; \delta_i  )  \ra   \overrightarrow{D_{H_n}\delta_n },\Psi $.


    \end{scriptsize}
\end{itemize}

After apply $(L\lor)$ for finite many times, we can obtain the desired result.

Next, in the proof of (iii), we consider the last rules applied in the derivation.
If the last rule is $(D^{+}_{T})$, we need to consider the cases:
when the principal formula $D_G A $ appears in $\Pi$, we need to apply the height-preserving invertibility in Proposition \ref{prop:inverti of t plus rule} and contraction rules;
when the principal formula $D_G A $ appears in $\Gamma$, the proof can be obtained directly from the definition.


When the last rule is $(D^{+}_{K})$, we assume that the derivation is ended in the form of 
$\Theta,\Sigma|\Pi, \Gamma \ra \Delta, \Lambda$ and 
 $p\notin \mathsf{V}(\Pi\cup\Lambda\cup\Theta)$, $a\notin \mathsf{Agt}(\Pi\cup\Lambda\cup\Theta)$.
The arguments are divided into the following cases:
\begin{enumerate}






\item  

 The right principal formula $D_H \alpha$  is in the multiset $\Lambda$.

 \begin{center}

\AxiomC{$\emptyset|\overrightarrow{\pi_m},\overrightarrow{\gamma_n} \ra \alpha$}
\RightLabel{$(D_{K}^+) $}
\UnaryInfC{$ \Theta',\overrightarrow{D_{P_m}\pi_m},\Sigma',\overrightarrow{D_{G_n}\gamma_n}| \Pi, \Gamma \ra
D_H \alpha, \Lambda',  \Delta$}
\DisplayProof

\end{center}
with $ \bigcup_{\overrightarrow{P_m}}\cup \bigcup_{\overrightarrow{G_n}} \subseteq H
$,  
where $\Theta',\overrightarrow{D_{P_m}\pi_m}$ is $\Theta$,  $\Sigma',\overrightarrow{D_{G_n}\gamma_n}$ is $\Sigma$
and $D_H \alpha, \Lambda'$ is $\Lambda$.

 
    In this case, $a\notin H$ since $ a\notin \mathsf{Agt}(\Lambda)$.

        \begin{enumerate}
        \item Some formulas in $\Sigma$ are principal, that is $\overrightarrow{\gamma_n}$ is not empty.
        In this case, we obtain that  $p$ and $a$  do not occur
         in formulas $\overrightarrow{\pi_m},\alpha$ from assumption. After applying induction 
        hypothesis, we obtain that:  $\emptyset|\overrightarrow{\pi_m} \ra 
          \mathcal{A}_{(p,a) }(\emptyset|\overrightarrow{\gamma_n} ;   \emptyset ),\alpha$  is derivable.

After applying $(L\neg) $, $(D_{K}^+) $, $(R\neg) $, weakening rules and $(D_{T}^+) $,
then $\emptyset | \Theta',\overrightarrow{D_{P_m}\pi_m}, \Pi \ra
 \langle D_{G_i}\rangle  \mathcal{A}_{(p,a) } (\emptyset|\overrightarrow{\gamma_n} ;    \emptyset ) ,
D_{H} \alpha, \Lambda'$ is derivable for some $i$ among $n$. 
Finally, we apply $(R\lor)$ for finite many times.


        

         \item All formulas in $\Gamma$ are not principal.  That is  $\overrightarrow{\gamma_n}$ is  empty.  We can derive the desired result by applying $(D_{K}^+) $ and weakening rules.

    \end{enumerate}







        \item The right principal formula $D_H \alpha$  is in the multiset $\Delta$, and $a\notin H$.

 \begin{center}
\AxiomC{$\emptyset|\overrightarrow{\pi_m},\overrightarrow{\gamma_n} \ra \alpha$}
\RightLabel{$(D_{K}^+) $}
\UnaryInfC{$ \Theta',\overrightarrow{D_{P_m}\pi_m},\Sigma',\overrightarrow{D_{G_n}\gamma_n}| \Pi, \Gamma \ra
D_H \alpha,  \Delta',\Lambda$}
\DisplayProof

\end{center}
with $ \bigcup_{\overrightarrow{P_m}}\cup \bigcup_{\overrightarrow{G_n}} \subseteq H
$, 
where   $\Theta',\overrightarrow{D_{P_m}\pi_m}$ is $\Theta$,  $\Sigma',\overrightarrow{D_{G_n}\gamma_n}$ is $\Sigma$
and $D_H \alpha, \Delta'$ is $\Delta$.
From assumption, $p$ and $a$ are not in $\overrightarrow{\pi_m}$. By induction hypothesis, we derive
$\emptyset|\overrightarrow{\pi_m} \ra  \mathcal{A}_{(p,a) }(\emptyset|\overrightarrow{\gamma_n};\alpha)$. 
After applying   $(D_{K}^+) $,  weakening rules and $(D_{T}^+) $,
then $ \emptyset|\Theta',\overrightarrow{D_{P_m}\pi_m}, \Pi \ra
D_{H}   \mathcal{A}_{(p,a) } (\emptyset|\overrightarrow{\gamma_n};\alpha ) ,
 \Lambda$ is derivable. 
Finally, we apply $(R\lor)$ for finite many times.

       \item The right principal formula $D_H \alpha$ is in the multiset $\Delta$, and $a\in H$.


 \begin{center}
\AxiomC{$\emptyset|\overrightarrow{\pi_m},\overrightarrow{\gamma_n} \ra \alpha$}
\RightLabel{$(D_{K}^+) $}
\UnaryInfC{$ \Theta',\overrightarrow{D_{P_m}\pi_m},\Sigma',\overrightarrow{D_{G_n}\gamma_n}| \Pi, \Gamma \ra
D_H \alpha,  \Delta',\Lambda$}
\DisplayProof

\end{center}
with $ \bigcup_{\overrightarrow{P_m}}\cup \bigcup_{\overrightarrow{G_n}} \subseteq H
$, 
where   $\Theta',\overrightarrow{D_{P_m}\pi_m}$ is $\Theta$,  $\Sigma',\overrightarrow{D_{G_n}\gamma_n}$ is $\Sigma$
and $D_H \alpha, \Delta'$ is $\Delta$.

Depending on whether $H$ is equal to $\{a\}$, we distinguish the following subcases.

\begin{enumerate}
    \item  When $H\ne \{a\}$,   $p$ and $a$  
        does not occur in formulas $\overrightarrow{\pi_m}$ from assumption. Then, according to induction hypothesis, 
        $\emptyset|\overrightarrow{\pi_m} \ra  \mathcal{A}_{(p,a) }(\emptyset|\overrightarrow{\gamma_n};\alpha)$ is derivable.
After applying   $(D_{K}^+) $,  weakening rules and $(D_{T}^+) $,
then $ \emptyset|\Theta',\overrightarrow{D_{P_m}\pi_m} \Pi \ra
D_{H/\{a\}}   \mathcal{A}_{(p,a) } (\emptyset|\overrightarrow{\gamma_n};\alpha) ,
 \Lambda$ is derivable. 
Finally, we apply $(R\lor)$ for finite many times.



 
\item
When $H= \{a\}$, we have $ \emptyset|\overrightarrow{\gamma_n} \ra \alpha$ is derivable, then we can apply 
the line 11  in Definition \ref{dfn:ap formula for t}.   It is obvious that $ \vdash \Theta|\Pi\ra \top, \Lambda.$ \qedhere

    

\end{enumerate}

\end{enumerate}
\end{proof}

Then, we can transfer the above results of $\gktdplus$ to $\gktd$.

\begin{cor}
\label{cor: cor main theorem of gktn}

    Let $\Gamma,\Delta$ be finite multi-sets of formulas. For every propositional variable $p$, agent symbol $a$, there exists a formula $\mathcal{A}_{(p,a)}(\Gamma;\Delta)$ such that:
    \begin{enumerate}[(i)]
        \item  
            $\mathsf{V}( \mathcal{A}_{(p,a)}(\Gamma;\Delta)) \subseteq \mathsf{V}(\Gamma\cup \Delta)\backslash\{p\}    $,
            $\mathsf{Agt}( \mathcal{A}_{(p,a)}(\Gamma;\Delta)) \subseteq \mathsf{Agt}(\Gamma\cup \Delta)\backslash\{a\}    $
        \item $\gktd\vdash    \Gamma , \mathcal{A}_{(p,a)}(\Gamma; \Delta) \ra \Delta$
        \item given finite multi-sets $\Pi, \Lambda$ of formulas such that 
            $p\notin \mathsf{V}(\Pi\cup\Lambda)$,$a\notin \mathsf{Agt}(\Pi\cup\Lambda)$ and $\gktd\vdash  \Pi,\Gamma \ra \Delta, \Lambda$
        then 
            $\gktd\vdash  \Pi \ra  \mathcal{A}_{(p,a)}(\Gamma; \Delta) ,\Lambda$
      
    \end{enumerate}
    
\end{cor}
\begin{proof}
    Let $\mathcal{A}_{(p,a)}(\Gamma;\Delta) $ be $\mathcal{A}_{(p,a)}(\emptyset|\Gamma;\Delta) $. Then, we apply Lemma \ref{lem:equipplloent of gktn and gktnplus} and Theorem \ref{thm:main theorem of gktnplus}. \end{proof}

    By  Theorem \ref{thm:main theorem of gkn} and Corollary \ref{cor: cor main theorem of gktn}, we can directly obtain the UIP for the
post-interpolant formula (i.e., $\mathcal{A}$-formula) in $\gkd,\gkdd$ and $\gktd$ with a single agent symbol and a single propositional variable.
For the next step, we need to show that we can derive UIP for both pre-interpolant and post-interpolant formulas   with multiple agent symbols and propositional variables.


\begin{defn}
   Let $p$ be a propositional variable, $a$ be an agent symbol and $\beta$ be a formula. We define $\mathcal{A}_{(p,a)}(\beta)$ as $\mathcal{A}_{(p,a)}(\emptyset;\beta)$.
    Furthermore, we define $\mathcal{E}_{(p,a)}(\beta)$ as $\neg\mathcal{A}_{(p,a)}(\neg \beta)$, namely $\neg \mathcal{A}_{(p,a)}(\emptyset;\neg \beta)$.
\end{defn}

In the following contents, the formula
 $\alpha(\overrightarrow{p_n},\overrightarrow{a_m})$ 
denotes all occurrences  of propositional variables and agent symbols in the formula $\alpha$.
\begin{cor}
\label{cor:uip in k, kd, kt multiple}
 Let $\mathbf{L}\in \{ \mathbf{K}_D, \mathbf{KD}_D,\mathbf{KT}_D\}$.   Uniform interpolation properties are satisfied in $\mathsf{G}(\mathbf{L})$.
For any formula $\alpha(\overrightarrow{p}, \overrightarrow{q},\overrightarrow{a},\overrightarrow{b})$, such that all $p$ and $q$ are distinguished propositional variables, all $a$ and $b$ are distinguished agent symbols,
 there exists a formula (pre-interpolant) $\mathcal{I}_{pre(\overrightarrow{q},\overrightarrow{b})} (\alpha)   $  such that:
    1, all $ \overrightarrow{q}$ and  $ \overrightarrow{b}$ do not occur in 
    $\mathcal{I}_{pre(\overrightarrow{q},\overrightarrow{b})} (\alpha) $;
    2, $\mathcal{I}_{pre(\overrightarrow{q},\overrightarrow{b})} (\alpha)  \ra \alpha$
    is derivable in $\mathsf{G}(\mathbf{L})$;
    3, for any formula $\beta(\overrightarrow{p},\overrightarrow{r},\overrightarrow{a},\overrightarrow{c})$, where
   all $p,q,r$ are distinguished propositional variables, all $a,b,c$   are distinguished agent symbols,
    if $\beta  \ra \alpha$ is derivable in $\mathsf{G}
    (\mathbf{L})$ then   
$\beta
    \ra
\mathcal{I}_{pre(\overrightarrow{q},\overrightarrow{b})} (\alpha)
    $ is derivable in $\mathsf{G}(\mathbf{L})$.
   Furthermore,  
for any formula $\alpha(\overrightarrow{p}, \overrightarrow{q},\overrightarrow{a},\overrightarrow{b})$, such that all $p$ and $q$ are distinguished propositional variables, all $a$ and $b$ are distinguished agent symbols,
 there exists a formula (post-interpolant) $\mathcal{I}_{post(\overrightarrow{q},\overrightarrow{b})} (\alpha)  $  such that:
   1, all $ \overrightarrow{q}$ and  $ \overrightarrow{b}$ do not occur in   $\mathcal{I}_{post(\overrightarrow{q},\overrightarrow{b})} (\alpha) $; 
    2, $\alpha\ra\ \mathcal{I}_{post(\overrightarrow{q},\overrightarrow{b})} (\alpha)   $ is derivable in $\mathsf{G}(\mathbf{L})$;
    3,  for any formula $\beta(\overrightarrow{p},\overrightarrow{r},\overrightarrow{a},\overrightarrow{c})$, where
   all $p,q,r$ are distinguished propositional variables, all $a,b,c$   are distinguished agent symbols,
    if $\alpha  \ra  \beta $ is derivable in $\mathsf{G}
    (\mathbf{L})$ then   $\mathcal{I}_{post(\overrightarrow{q},\overrightarrow{b})} (\alpha)\ra \beta  $ is derivable in $\mathsf{G}
    (\mathbf{L})$.
\end{cor}

\begin{proof}
First, we show the case of pre-interpolatant. 
Let $m,n,l\in \mathtt{N}$.
 Given an arbitrary formula $\alpha(\overrightarrow{p},\overrightarrow{q},\overrightarrow{a},\overrightarrow{b})$.
 The pre-interpolant formula is defined as follows:

\begin{equation}
\nonumber
\mathcal{I}_{pre(\overrightarrow{q},\overrightarrow{b})} (\alpha)  = 
    \begin{cases}
        \mathcal{A}_{(q_1,b_1 )} (    \cdots (\mathcal{A}_{(q_m,b_n )} (\alpha)) \cdots ), & \text{if } m =n. \\
        \mathcal{A}_{(q_1,b_1 )} (    \cdots  (\mathcal{A}_{(q_{n+l-1},b_n )} ( \mathcal{A}_{(q_{n+l},b_n )} (\alpha))) \cdots ) , & \text{if}~ m=n+l~\text{for}~l\neq0. \\
         \mathcal{A}_{(q_1,b_1 )} (    \cdots (\mathcal{A}_{(q_m,b_{m+l-1} )} (\mathcal{A}_{(q_m,b_{m+l} )} (\alpha))) \cdots ) , & \text{if}~ n=m+l~\text{for}~l\neq0.
    \end{cases}
\end{equation}
The second and third condition state that if $m\neq n$, we 
repeat the  occurrences of the same variable or agent symbol 
until all symbols have been dealt with.

Next, in the case of post-interpolant.  
Given an arbitrary formula $\alpha(\overrightarrow{p},\overrightarrow{q},\overrightarrow{a},\overrightarrow{b})$.
The post-interpolant formula is defined similarly as follows:

\begin{equation}
\nonumber
\mathcal{I}_{post(\overrightarrow{q},\overrightarrow{b})} (\alpha)  = 
    \begin{cases}
        \mathcal{E}_{(q_1,b_1 )} (    \cdots (\mathcal{E}_{(q_m,b_n )} (\alpha)) \cdots ), & \text{if } m =n. \\
        \mathcal{E}_{(q_1,b_1 )} (    \cdots  (\mathcal{E}_{(q_{n+l-1},b_n )} ( \mathcal{E}_{(q_{n+l},b_n )} (\alpha))) \cdots ) , & \text{if}~ m=n+l~\text{for}~l\neq0. \\
         \mathcal{E}_{(q_1,b_1 )} (    \cdots (\mathcal{E}_{(q_m,b_{m+l-1} )} (\mathcal{E}_{(q_m,b_{m+l} )} (\alpha))) \cdots ) , & \text{if}~ n=m+l~\text{for}~l\neq0.
    \end{cases}
\end{equation}

The conditions can be checked straightforward.
\end{proof}

 \section{Conclusion and Future Direction}
This paper provides the purely syntactic proof of UIP in epistemic logic with distributed knowledge
 $\mathbf{K}_D$, $\mathbf{KD}_D$  (in Theorem \ref{sec:main thm of gkd and gkdd}) , $\mathbf{KT}_D$ (in Corollary \ref{cor: cor main theorem of gktn}). In Corollary \ref{cor:uip in k, kd, kt multiple}, we show 
uniform interpolation properties for both pre-interpolant and post-interpolant formulas   with multiple agent symbols and propositional variables in these systems.
 Furthermore, not only propositional variables but also agent symbols are taken into the common language of the interpolant formula.

In the next step, it is interesting to show the UIP for distributed knowledge in modal logic $\mathbf{S5}$.
One well-known difficulty is that, the sequent calculus for  $\mathbf{S5}$ 
does not satisfy the cut elimination theorem (cf. \cite{takano92}). As a result, the  method in this paper cannot be directly applied.
In this direction, B{\'\i}lkov{\'a}, Fussner and Kuznets \cite{Bilkova2025AgentIF} have applied a method based on a combination of the features of hypersequents and nested sequents, and already proved the Lyndon (and hence Craig) interpolation property of multi-agent modal logic $\mathbf{S5}$.
Another direction is to  prove UIP in the intuitionistic modal system. To do this, we may need to make use of the $\mathbf{G4}$-style sequent calculus in Pitts \cite{pitts1992interpretation}.
Some progress have been made in this direction
\cite{Giessen2020ProofTF,Iemhoff2022TheGA,Feree2024}.









\bibliography{bib_sl.bib}

\end{document}